\documentclass[preprint,12pt,3p]{elsarticle}

\usepackage{times}
\usepackage{graphicx}
\usepackage{latexsym}
\usepackage{amssymb}

\usepackage{amsthm}
\usepackage{graphicx}
\usepackage{amsmath}
\usepackage{hyperref}
\usepackage{tabularx}
\usepackage{comment}
\usepackage{amsmath,graphicx}

\newtheorem{question}{\em Question}
\newtheorem{proposition}{\em Proposition}
\newtheorem{problem}{\em Problem}

\newtheorem{theorem}{\em Theorem}

\newtheorem{definition}{\em Definition}
\newtheorem{lemma}{\em Lemma}
\newtheorem{remark}{\em Remark}
\newtheorem{corollary}{\em Corollary}

\newcommand{\sw}{{\sf SW}}

\journal{Sample Journal}

\begin{document}

\begin{frontmatter}

\title{Assigning tasks to agents under time conflicts: a parameterized complexity approach\tnoteref{label0}}
\tnotetext[label0]{This work has been partially supported by the Italian MIUR PRIN 2017 Project ALGADIMAR "Algorithms, Games, and Digital Markets".}

\author[label1,label2]{Alessandro Aloisio\corref{cor1}}
\address[label1]{Universit\`a degli Studi dell'Aquila, L'Aquila, Italy}
\address[label2]{Gran Sasso Science Institute, L'Aquila, Italy\fnref{label4}}

\cortext[cor1]{I am corresponding author}

\ead{alessandro.aloisio@univaq.it}

\author[label2]{Vahan Mkrtchyan}
\ead{vahan.mkrtchyan@gssi.it}


\begin{abstract}
We consider the problem of assigning tasks to agents under time conflicts, with applications also to frequency allocations in point-to-point wireless networks. In particular, we are given a set $V$ of $n$ agents, a set $E$ of $m$ tasks, and $k$ different time slots. Each task can be carried out in one of the $k$ predefined time slots, and can be represented by the subset $e\subseteq E$ of the involved agents. Since each agent cannot participate to more than one task simultaneously, we must find an allocation that assigns non-overlapping tasks to each time slot. Being the number of slots limited by $k$, in general it is not possible to executed all the possible tasks, and our aim is to determine a solution maximizing the overall social welfare, that is the number of executed tasks. We focus on the restriction of this problem in which the number of time slots is fixed to be $k=2$, and each task is performed by exactly two agents, that is $|e|=2$. In fact, even under this assumptions, the problem is still challenging, as it remains computationally difficult. We provide parameterized complexity results with respect to several reasonable parameters, showing for the different cases that the problem is fixed-parameter tractable or it is paraNP-hard.
\end{abstract}


\end{frontmatter}



\section{Introduction}
\label{sec:intro}

In this paper we consider a problem of specific task-assignment to agents with time conflicts. 
Each agent belongs to a set $V$ of cardinality $n$, while a task is completely defined by the subset $e$ of agents needed to execute it. Since there are only $k$ time slots for carrying out the tasks, we want to find the best tasks-agents assignment in order to maximize the executed task, i.e., the ones assigned to a time slot.


Since each agent cannot participate to more than one task simultaneously, we must find an allocation that assigns non-overlapping tasks to each time slot. Being the number of slots limited by $k$, in general it is not possible to executed all the possible tasks, and our aim is to determine a solution maximizing the overall social welfare, that is the number of executed tasks. Since it is a first study, we focus on the restriction of this problem in which the number of time slots is fixed to be $k=2$, and each task is performed by exactly two agents, that is $|e|=2$. In fact, even under these assumptions, the problem is still challenging, as it remains computationally difficult. We provide parameterized complexity results with respect to several reasonable parameters, showing for the different cases that the problem is fixed-parameter tractable or it is paraNP-hard. The results are reported in Table~\ref{tab:1}.

With these restrictions, the problem can be described with a graph $G(V,E)$, where nodes represents agents, and each edge represent a task assigned to a couple of agents. In the general version, the graph becomes a hypergraph, and each task can be executed by more than two agents.  

Parameterized complexity theory has been raising more interest in the last few years because it looks at giving a finer running time analysis with respect to the classical complexity theory. This allows to better understand where the difficulty of a problem lies. The main idea is to provide algorithms whose running times are based not only on the input size but also on one or more parameters. In particular, this kind of algorithms require that the parameters do not appear in the exponent part of the input size so that the running time is exponential only in the parameters and not in the input size. More details on parameterized complexity theory can be found in \cite{FPTbook}.

The rest of this paper is organized as follows. The problem description is formally stated in Section~\ref{sec:prob_def}, where we also relate the problem to the maximum $k$-edge-colorable subgraph problem. In Section~\ref{sec:motwork_tasks}, we provide some related work to the assigning problems and to the edge coloring one. In Sections~\ref{sec:aux}, we present some auxiliary results that are useful in the rest of the paper. The first part of our results are described in Section \ref{sec:main}. In Section~\ref{sec:dynamic}, we focus on two dynamic programming algorithm for branchwidth, treewidth, and cliquewidth.  In Section \ref{sec:IterativeCompression}, we consider two problems closely-related to the maximum 2-edge-colorable subgraph problem. Using the method of iterative compression, we show that they are FPT with respect to budget $k$.  We conclude the paper in Section \ref{sec:conc}, where we summarize our results and outline the directions for future work.

\section{Problem definition}\label{sec:prob_def}
We provide here a formal definition of the assigning tasks to agents under time conflicts problem ($T2ATC(k)$), where $k$ is the number of time slots, by using hypergraphs.
For a hypergraph $\mathcal{H}=(V,E)$, we denote by $V$ its node set of cardinality $n$, and by $E$ its edge set of cardinality $m$.
Unless otherwise stated, $\mathcal{H}$ 
is assumed to be undirected, and without multiple edges or self-loops.

Each node $u\in V$ corresponds to an agent, while each hyperedge $e\in E$ corresponds to a task, where the nodes in $e$ are the agents involved in the task. An edge can be assigned to one of the $k$ time slots available. Let $[k]$ denote $\{0,1,\ldots k\}$. If we consider also the dummy time slot $0\in [k]$ that essentially means 'not assigned', then we can say that every edge is assigned to a shift. The addition of the dummy time slot is useful for the problem definition and for some of the theorems.  

The utility of an agent $u\in V$ is defined as the number of the tasks she belongs to that are assigned to a non dummy time slot. That is $\mu_u(\tau)=\sum_{e\in E\colon u\in e}w(e)$, where $\tau$ is a global assignment of the tasks to the time slots in $[k]$, and $w(e)\in \{0,1\}$ is the weight of edge $e$ that defines if $e$ is assigned to a real or a to the dummy shift. The problem definition is the following one.

\par\medskip\noindent
\begin{tabularx}{\columnwidth}{@{\hspace{1mm}}l@{\hspace{1mm}}X@{\hspace{1mm}}}
	\hline & \\[-4mm]
	\multicolumn{2}{c}{\hspace{-2mm}\textsf{T2ATC(p) : Ass. tasks to agents under time conflicts}\hspace*{-2mm}}\\
	\hline \ & \\[-3mm]
	\bfseries{\textit{Input}:} & A hypergraph $\mathcal{H}=(V,E)$; and an integer $k\geq 2$.\\ 
	\bfseries{\textit{Solution}:} & An assignment of each task to a time slot $\tau \colon E \to [k]$ covering $\mathcal{H}$ such that for all $u\in V$, and for every two edges $e_1$ and $e_2$ that contain $u$, $e_1$ and $e_2$ belong to the same time slot if and only if $\tau(e_1)=\tau(e_2)=0$, that is both task $e_1$ and $e_2$ are not assigned; and a binary weight function $w\colon E \to \{0,1\}$ that equals 1 if a task is assigned to a non dummy time slot, 0 otherwise.\\
	\bfseries{\textit{Goal}:} & Maximize the Social Welfare $\sw(\tau)=\sum_{u\in V} \mu_u(\tau)=\sum_{u\in V}\sum_{e\in E\colon u\in e}w(e)$, i.e., the sum of all the agents' utilities.\\[0.5mm]	
	\hline
\end{tabularx}
\par\medskip\noindent

The social welfare of a particular assignment (solution) can be rewritten as  the sum of the weight of each hyperedge times the number of the involved agents (i.e. the cardinality of the hyperedge), that is $\sw(\tau)=\sum_{e\in E}|e|w(e)$. 

As already said in the introduction, in this paper we focus only on tasks assigned to exactly two agents, that is $|e|=2$. This allows us to represent the problem on usual graphs with usual edges. Under this restriction, we can use an equivalent metric for the social welfare, that is $\sw(\tau)=\sum_{e\in E}w(e)$, because all the edges have cadinality 2, so we can remove 2. For usual graphs, and with this new definition of social welfare, it is clear that the problem is equivalent to finding a $k$-edge-colorable subgraph with maximum number of edges together with its $k$-edge-coloring. In fact, we can see each non dummy time slot as a color that can be assigned to each edge, avoiding that a node has two edges with the same color incident to it. The dummy time slot, on the other hand, means 'not colored'. Because of this problem equivalence, all the results we provide are given for the edge coloring problem.    

\paragraph{Edge coloring problem definition}
Since we focus on graphs ($|e|=2$) and two time slots ($k=2$), we recall some useful results, and give a formal definition of the problem as $k$-edge-coloring.

The set of vertices and edges of a graph $G$ is denoted by $V$ and $E$, respectively, $d_{G}(u)$ denotes the degree of a vertex $u$ of $G$. Let $\delta(G)$ and $\Delta(G)$ be the minimum and maximum degree of vertices of $G$. Let $rad(G)$ and $diam(G)$ be the radius and diameter of $G$. 

A {matching} in a graph $G$ is a subset of $E$ such that no vertex of $G$ is incident to two edges from it. A {maximum matching} is a matching that contains the largest possible number of edges.

For $k\geq 0$, a graph $G$ is {$k$-edge colorable}, if its edges can be assigned colors from a set of $k$ colors so that adjacent
edges receive different colors. The smallest $k$, such that $G$ is $k$-edge-colorable is called chromatic index of $G$ and is denoted by $\chi'(G)$. The classical theorem of Shannon states that for any multi-graph $G$, $\Delta(G)\leq \chi'(G) \leq \left \lfloor \frac{3\Delta(G)}{2} \right \rfloor$ \cite{Shannon:1949,stiebitz:2012}. Moreover, the classical theorem of Vizing states that for any multi-graph $G$, $\Delta(G)\leq \chi'(G) \leq \Delta(G)+\mu(G)$ \cite{stiebitz:2012,vizing:1964}. Here $\mu(G)$ denotes the maximum multiplicity of an edge of $G$. A multi-graph $G$ is class I, if $\chi'(G)=\Delta(G)$, otherwise it is class II.

If $k<\chi'(G)$, we cannot color all edges of $G$ with $k$ colors. Therefore, it is natural to investigate the maximum number of edges that one can color with $k$ colors. A subgraph $H$ of $G$ is called {maximum $k$-edge-colorable}, if $H$ is $k$-edge-colorable and contains maximum number of edges among all $k$-edge-colorable subgraphs of $G$. For $k\geq 0$ and a graph $G$ let\\

$\nu_{k}(G) = \max \{ |E(H)| : H$  is a $k$-edge-colorable subgraph of  $G \}.$\\

Clearly, a $k$-edge-colorable subgraph is maximum if it contains exactly $\nu_k(G)$ edges. Observe that $\nu_1(G)$ is the size of a maximum matching of $G$. We will shorten this notation to $\nu(G)$. In this paper, we deal with the exact solvability of the maximum $k$-edge-colorable subgraph problem. Its precise formulation is the following:
\begin{problem}\label{prob:MaxkEdgeColsub}
	(Maximum $k$-edge-colorable subgraph) Given a graph $G$ and an integer $k$, find a $k$-edge-colorable subgraph with maximum number of edges together with its $k$-edge-coloring.
\end{problem} We investigate this problem from the perspective of fixed-parameter tractability. Recall that an algorithmic problem $\Pi$ is fixed-parameter tractable with respect to a parameter $\theta$, if there is an exact algorithm solving $\Pi$, whose running-time is $f(\theta)\cdot poly(size)$. Here $f$ is some (computable) function of $\theta$, $size$ is the length of the input and $poly$ is a polynomial function. A (parameterized) problem is paraNP-hard, if it remains NP-hard even when the parameter is constant.

In this paper, we focus on the maximum 2-edge-colorable subgraph problem which is the restriction of the problem to the case $k=2$. We present some results that deal with the fixed-parameter tractability of this problem with respect to various graph-theoretic parameters. The main contributions of this paper are the following:

\begin{itemize}
	\item ParaNP-hardness of the problem with respect to the radius, diameter and $|V|-MaxLeaf(G)$,
	\item Fixed-parameter tractability of our problem with respect to $|V|-\delta$, branchwidth, treewidth, the size of largest matching, the dimension of the cycle space and $MaxLeaf(G)$.
	\item Polynomial time solvability of our problem for the graphs of bounded cliquewidth.
	\item Fixed-parameter tractability of two related problems with respect to the budget $k$. 
	
\end{itemize}

The results obtained in this paper are summarized in Table~\ref{tab:1}. For the notions, facts and concepts that are not explained in the paper the reader is referred to \cite{FPTbook,west:1996}.

\begin{table*}[t]
	\caption{Summary of the main results obtained the paper about maximum 2-edge-colorable problem.}\label{tab:1}
	\footnotesize
	\begin{tabularx}{\textwidth}{|X|X|l|l|}
		\hline
		\textbf{Results} & \textbf{Parameter} & \textbf{in FPT ?} & \textbf{Time}\\
		\hline
		Theorem \ref{thm:radius} & 	Radius $rad(G)$  & paraNP-hard & \\
		\hline
		Remark \ref{rem:diam}& Diameter $diam(G)$ & paraNP-hard & \\
		\hline
		Theorem \ref{thm:HardnessMax2SubgraphMaxDegree}  & $\delta$, $\Delta$, $|V|-\Delta$, Number of maximum-degree vertices & paraNP-hard & \\
		\hline
		Proposition \ref{prop:Vminusdelta} & $|V|-\delta$ & Yes &  $2^{4(|V|-\delta)^2}$\\
		\hline
		Theorems \ref{thm:TreeWidth},\ref{thm:TreeWidth2} & Treewidth & Yes  &  $O(6^{\frac{4}{3}(h+1)}(2|V|-1))$\\
		\hline
		Theorem \ref{th:BranchWidth} & Branchwidth & Yes  &   $O(6^{2h}(2|V|-1))$\\
		\hline
		Theorem \ref{thm:CliqueWidth} & Cliquewidth & (?) &  $O(h^2 |V|^{(16h^2-14h-1)}(|V|+|E|))$\\
		\hline
		Corollary \ref{cor:FPTnu1} & Maximum matching $\nu(G)$ & Yes & $O(6^{\frac{4}{3}(h+1)}(2|V|-1))$ \\
		\hline
		Theorem \ref{thm:FPTdimCycSpace} & Dimension of the cycle space & Yes &  $3^{k}\cdot poly(size)$\\
		\hline
		
		Proposition \ref{prop:MaxLeafFPT} & $MaxLeaf(G)$ & Yes &  $2^{(MaxLeaf(G)+1)^2}$\\
		\hline
		
		Proposition \ref{prop:|V|minusMaxLeaf} & $|V|-MaxLeaf(G)$ & paraNP-hard &  \\
		\hline
	\end{tabularx}
\end{table*}
\normalsize

\section{Motivation and related work}
\label{sec:motwork_tasks}

In this section we report some motivation and related work concerning the assignment and the edge coloring problems defined in Section~\ref{sec:prob_def}. Assigning tasks to agents, and agents to task have been extensively investigated both in artificial intelligence, combinatorial optimization, operational research, and other scientific fields. Moreover, there exist many different variations of the problem. We report here just a small part of the literature that we think it could be useful to collocate our problem.

A standard one-to-one general assignment problem can be defined with a set of tasks $E$ and a set of agents $V$, and a utility matrix that gives the utility obtained by each agent when it is assigned to a specific task. The first paper talking about the assignment problem appeared in 1952 \cite{Votaw_Orden:1952}, while a survey can be found in \cite{Pentico:2007}. Starting from the original one-to-one problem, many different variations have been formulated and investigated. A famous one is the Generalized Assignment Problem \cite{Oncan:2007}.

Our problem, $T2ATC(k)$, is a particular version of assignment problems, because the hypergraph describes which agents can collaborate to carry out a specific task, that is a task cannot be done by any possible subset of agents. This can be useful, because in many real cases agents are not interchangeable. Moreover, also not every possible agent coalition is allowed (if the hypergraph is not complete). This means that our model is flexible and can be applied to different real problems. 

Another way to see $T2ATC(k)$ is to partition the tasks in time slots (including the dummy one), in order to maximize the social welfare, that is maximizing the number of tasks assigned to a non dummy time slot. This allows us to view our problem also as group structure formation (CSG) \cite{Rahwan:2015,flammini:2018,Rahwan:2008,aloisio_vinci:2020}. This dual view can be useful for frequency allocations in point-to-point wireless networks problems, when $|e|=2$, so when the problem is defined on usual graphs. In fact, each edge is a point-to-point link, $k$ is the number of available frequencies, and, clearly, two adjacent edges cannot use the same frequency, because of interference problems on the nodes, which represent wireless devices.     

\paragraph{Edge coloring related work}

Since we focus on graphs ($|e|=2$) and two time slots ($k=2$), we report also motivation and related works concerning the edge coloring problem. 

There are many papers where the ratio $\frac{\nu_k(G)}{|E|}$ has been investigated. \cite{bollobas:1978,henning:2007,nishizeki:1981,nishizeki:1979,weinstein:1974} prove lower bounds for this ratio in case of regular graphs and $k=1$. For regular graphs of high girth the bounds are improved in \cite{flaxman:2007}. Albertson and Haas investigated the problem in \cite{haas:1996,haas:1997} when $G$ is a cubic graph. See also \cite{samvel:2010}, where it is shown that for every cubic multigraph $G$, $\nu _{2}(G) \geq \frac{4}{5}|V|$ and $\nu _{3}(G) \geq  \frac{7}{6} |V|$. Moreover, \cite{samvel:2014} proves that for any cubic multigraph $G$, $\nu _{2}(G) + \nu _{3}(G) \geq 2|V|$, and in \cite{samvel:2010,corrigendum} Mkrtchyan et al. showed that for any cubic multigraph $G$, $\nu_{2}(G) \leq \frac{|V| + 2\nu_{3}(G)}{4}$. Finally, in \cite{LianaDAM:2019}, it is shown that the sequence $\nu_k$ is convex in the class of bipartite multigraphs. Rizzi in \cite{Rizzi:2009} has shown that the above-mentioned $\frac{7}{6} |V|$ bound for cubic multigraphs can be significantly improved for graphs (without parallel edges) $G$ of maximum degree three. For such graphs $G$, it can be shown that $\nu_3(G) \geq \frac{6}{7}\cdot |E|$ \cite{Rizzi:2009}.

Bridgeless cubic graphs that are not $3$-edge-colorable are called snarks \cite{cavi:1998}, and the ratio for snarks is investigated by Steffen in \cite{steffen:1998,steffen:2004}. This lower bound has also been investigated in the case when the graphs need not be cubic in \cite{miXumbFranciaciq:2013,Kaminski:2014,Rizzi:2009}. Kosowski and Rizzi have investigated the problem from the algorithmic perspective \cite{Kosowski:2009,Rizzi:2009}. The problem of finding a maximum $k$-edge-colorable graph in an input graph is NP-complete for every fixed $k\geq 2$. For example, when $G$ is cubic and $k=2$, we have that $\nu_2(G)=|V|$ if and only if $G$ contains two edge-disjoint perfect matchings. The latter condition is equivalent to saying that $G$ is 3-edge-colorable, which is an NP-complete problem as Holyer has demonstrated in \cite{holyer:1981}. Thus, it is natural to investigate the (polynomial) approximability of the problem. In \cite{FeigeOfekWieder} for each $k\geq 2$ an approximation algorithm for the problem is presented. There for each fixed value of $k \geq 2$, algorithms are proved to have certain approximation ratios and these ratios are tending to $1$ as $k$ goes to infinity. In \cite{Kosowski:2009}, two approximation algorithms for the maximum 2-edge-colorable subgraph and maximum 3-edge-colorable subgraph problems are presented whose performance ratios are $\frac{5}{6}$ and $\frac{4}{5}$, respectively. Finally, note that the results of \cite{FeigeOfekWieder} are improved for $k=3,...,7$ in \cite{Kaminski:2014}.

Some structural properties of maximum $k$-edge-colorable subgraphs of graphs are proved in \cite{samvel:2014,MkSteffen:2012}. There it is shown that every set of disjoint cycles of a graph with $\Delta=\Delta(G) \geq 3$ can be extended to a maximum $\Delta(G)$-edge colorable subgraph. Moreover, there it is shown that any maximum $\Delta(G)$-edge colorable subgraph of a graph is always class I. Observe that this statement is not true when $G$ is a multigraph. If one considers a triangle in which each edge is of multiplicity three, then the maximum degree in it is six. An example of a maximum 6-edge-colorable in this graph will be the triangle in which each edge is of multiplicity two. Observe that it has maximum degree four and chromatic index six. Thus, it is class II. Finally, in \cite{MkSteffen:2012} it is shown that if $G$ is a graph of girth (the length of the shortest cycle) $g \in \left \{ 2k, 2k+1 \right \} (k \geq 1)$ and $H$ is a maximum $\Delta(G)$-edge colorable subgraph of $G$, then $\frac{|E(H)|}{|E|} \geq \frac{2k}{2k+1}$, The bound is best possible as there is an example attaining it.

In \cite{kEdgeColoringFPT} the $k$-edge-coloring problem is considered, which is formulated as follows:
\begin{problem}
	\label{prob:kEdgeColoring} ($k$-edge-coloring) Given a graph $G$ and an integer $k$, check whether $G$ is $k$-edge-colorable.
\end{problem} There it is shown that for each fixed $k$, the $k$-edge-coloring problem is fixed-parameter tractable with respect to the number of maximum degree vertices of the input graph. Observe that the maximum $k$-edge-colorable subgraph problem is harder than $k$-edge-coloring, as if we can construct a maximum $k$-edge-colorable subgraph $H_k$ of the input graph $G$, then in order to see that whether $G$ is $k$-edge-colorable, we just need to check whether $E(H_k)=E$. If one considers the edge-coloring problem, where for an input graph $G$, we need to find a $\chi'(G)$-edge-coloring of $G$, then in \cite{KowalikSIDMA:2018} it is stated that a major challenge in the area is to find an exact algorithm for this problem whose running-time is $2^{O(n)}=O(c^n)$. Observe that the maximum $k$-edge-colorable subgraph problem is harder than edge-coloring. If we are able to solve the maximum $k$-edge-colorable subgraph problem in time $O(f(size))$, then we can solve the Edge-Coloring problem in time $O(f(size))\cdot \log(|V|)$. In order to see this, just observe that we can do a binary search on $k=1,2,...,|V|$, solve the maximum $k$-edge-colorable problem and find an edge-coloring of $G$ with the smallest number of colors. Here we used the fact that any graph $G$ is $|V|$-edge-colorable.

\section{Some auxiliary results}
\label{sec:aux}

In this section, we present some results that will be used in obtaining the main results of the paper. Below we assume that $\mathbb{N}$ is the set of natural numbers.

\begin{lemma}
	\label{lem:Redparam} (\cite{Sasak:2010}) Let $\Pi$ be an algorithmic problem, and let $k_1$ and $k_2$ be some parameters. Assume that there is a (computable) function $g:\mathbb{N}\rightarrow \mathbb{N}$, such that for any instance $I$ of $\Pi$, we have $k_1(I)\leq g(k_2(I))$. Then if $\Pi$ is FPT with respect to $k_1$, then it is FPT with respect to $k_2$.
\end{lemma}


In \cite{holyer:1981}, Holyer has shown that checking whether a cubic graph is 3-edge-colorable is an NP-complete problem. For a cubic graph $G$, let $r_3(G)$ be defined as:
\[r_3(G)=|E|-\nu_3(G).\]
This parameter is introduced and investigated in \cite{steffen:2004}. In particular, there it is observed there that $r_3(G)\neq 1$ for any cubic graph $G$. This means that $r_3(G)$ can be zero or at least two, and the 3-edge-coloring problem in cubic graphs amounts to deciding which of these two cases holds. For our purposes we will consider the following restriction of 3-edge-coloring problem in cubic graphs:\\

{\bf Problem 3:} For a fixed integer $l\geq 1$, consider a decision problem, whose input is a cubic graph $G$, in which $r_3(G)$ is from the set $\{0,l,l+1, l+2,...\}$. The goal is to check whether $G$ is 3-edge-colorable, that is, whether $r_3(G)=0$.

\begin{lemma}
	\label{lem:AuxProb1} For each fixed $l\geq 1$, Problem 3 is $NP$-complete.
\end{lemma}

\begin{proof} The case when $l\leq 2$ corresponds to the usual 3-edge-coloring problem in cubic graphs. Thus, we can assume that $l\geq 3$. We reduce the 3-edge-coloring problem of cubic graphs to this problem. Let $G$ be any cubic graph. Consider a cubic graph $H$ obtained from $l$ vertex disjoint copies of $G$. Observe that $|V(H)|=l\cdot |V|$, hence $H$ can be constructed from $G$ in linear time. Now, it is easy to see that $G$ is 3-edge-colorable if and only if $H$ is 3-edge-colorable. Moreover, $r_3(H)=l\cdot r_3(G)$. Hence, $r_3(H)$ is either zero or at least $l$. The proof is complete.
\end{proof}

We will also need the following result obtained in \cite{samvel:2010,corrigendum}:

\begin{theorem}
	\label{thm:SamvelInEq} For any cubic graph $G$ $\nu_{2}(G) \leq \frac{|V| + 2\nu_{3}(G)}{4}$.
\end{theorem}


%
%

\section{First part of theoretical results}
\label{sec:main}

In this section, we present the first part of our main results about the maximum 2-edge-colorable subgraph problem. If $m$ is the number of edges of the input graph $G$, then clearly we can generate all $2^m$ subgraphs/subsets of $E$, and check each of them for $2$-edge-colorability. In great contrast with $k$-edge-colorability with $k\geq 3$, checking 2-edge-colorability can be done in polynomial time. A subgraph $F$ of $G$ is 2-edge-colorable if and only if it has maximum degree at most two, and it contains no component that is an odd cycle. Clearly this can be checked in polynomial time. The running time of this trivial, brute-force algorithm is $O^{*}(2^m)$. We will refer to this algorithm as trivial or brute-force algorithm.


The first parameter with respect to which we will investigate our problem is the radius of the graph.

\begin{theorem}
	\label{thm:radius} The maximum 2-edge-colorable subgraph problem is paraNP-hard with respect to the $rad(G)$.
\end{theorem}

\begin{proof} We present a reduction from Problem 3 with $l\geq 6$. By Lemma \ref{lem:AuxProb1} it is $NP$-complete. Let us take an arbitrary cubic graph $G$ with $r_3(G)$ either zero or at least $l$. Take a new vertex $z$, who is joined to every vertex of $G$. Let $G'$ be the resulting graph.
	
	Let us show that $\nu_2(G')\geq |V|$ if and only if $G$ is 3-edge-colorable. Let $G$ be a 3-edge-colorable. Then it admits a pair of edge-disjoint perfect matchings. Hence, these perfect matchings form a 2-edge-colorable subgraph in $G'$. Thus, $\nu_2(G')\geq |V|$. Now, assume that $G$ is not 3-edge-colorable, hence $r_3(G)\geq l\geq 6$. By Theorem \ref{thm:SamvelInEq}:
	
	\begin{align*}
	&\nu_2(G')\leq 2+ \nu_2(G)\leq 2+ \frac{|V|+2\cdot \nu_3(G)}{4}=2+|V|-\frac{r_3(G)}{2}\leq |V|-1,
	\end{align*}
	
	since $r_3(G)\geq l\geq 6$. Hence, if $\nu_2(G')\geq |V|$, then $G$ is 3-edge-colorable.
	
	Observe that in graphs $G'$ that we obtained from $G$, we have $rad(G')=1$ ($z$ is of distance one from any other vertex). Thus, checking whether $\nu_2(G')=|V|$ is an NP-complete problem even when the radius is one. Thus the problem is paraNP-hard with respect to the radius. The proof is complete.
\end{proof}

\begin{remark}
	\label{rem:diam} The maximum 2-edge-colorable subgraph problem is paraNP-hard with respect to the $diam(G)$.
\end{remark} This follows from Theorem \ref{thm:radius}, Lemma \ref{lem:Redparam} and the fact that in any graph $G$, we have 
\[rad(G)\leq diam(G)\leq 2\cdot rad(G).\]

In \cite{kEdgeColoringFPT}, it is shown that for each fixed $k$, the $k$-edge-coloring problem is FPT with respect to the number of maximum degree vertices of the input graph. As we have mentioned previously, the maximum $k$-edge-colorable subgraph problem is harder than $k$-edge-coloring. Thus, one can try to parameterize the latter with respect to the number of vertices of maximum degree. As the following theorem states, if $P\neq NP$, this is impossible.

\begin{theorem}
	\label{thm:HardnessMax2SubgraphMaxDegree} The maximum 2-edge-colorable subgraph problem is paraNP-hard with respect to the number of maximum-degree vertices.
\end{theorem}

\begin{proof} Consider the class of graphs $G'$ from the proof of Theorem \ref{thm:radius}. Observe that if $G$ is the complete graph on four vertices then $G'$ has five vertices of degree four, which have maximum degree in $G'$. On the other hand, if $|V|\geq 6$, then $z$ is the only vertex of maximum degree. Thus, the problem is NP-hard when the number of maximum degree vertices is at most five. The proof is complete.
\end{proof}

\begin{remark}
	\label{rem:NotAllMaxDegrees} Observe that in the above proof, there is no need for us to join $z$ to all the vertices of $G$. Since $G$ is cubic we can join $z$ to five vertices of $G$. This will lead to the graph $G'$, where $z$ is the only vertex of degree five, which is maximum. All other vertices are of degree four or three. Thus, the problem remains hard even when the number of maximum degree vertices is one and the maximum degree is five.
\end{remark}

Holyer's result \cite{holyer:1981} implies that it is $NP$-hard to find a maximum 2-edge-colorable subgraph in cubic graphs. Thus, the maximum 2-edge-colorable subgraph problem is paraNP-hard with respect to $\Delta(G)$ and $\delta(G)$. Moreover, in the proof of Theorem \ref{thm:radius}, we have $|V(G')|=|V|+1$ and $\Delta(G')=d(z)=|V|$, hence $|V(G')|-\Delta(G')=1$ in these graphs $G'$. Thus, one can say that it is paraNP-hard with respect to $|V|-\Delta$, too. On the positive side, it turns out that
\begin{proposition}
	\label{prop:Vminusdelta} The maximum 2-edge-colorable subgraph problem is FPT with respect to $|V|-\delta$.
\end{proposition}

\begin{proof} Let $G$ be any graph. If $|V|-\delta(G)\geq \frac{|V|}{2}$, then 
	\[|V|\leq 2\cdot (|V|-\delta(G)).\]
	Thus,
	\[|E|\leq |V|^2\leq 4\cdot (|V|-\delta(G))^2.\]
	Now, if we run the trivial algorithm, its running-time will depend solely on $|V|-\delta$, as we have bounded the number of edges in terms of it. On the other hand, if $|V|-\delta(G)\leq \frac{|V|}{2}$, then 
	\[\delta(G)\geq \frac{|V|}{2}.\]
	Thus, by Ore's classical theorem \cite{west:1996}, $G$ has a Hamiltonian cycle $C$. Now, if $|V|$ is even, then $C$ is a 2-edge-colorable subgraph in $G$. Since in any graph $G$, $\nu_2(G)\leq |V|$, we have that $C$ is a maximum 2-edge-colorable subgraph in $G$. On the other hand, if $|V|$ is odd, then any matching in $G$ has at most $\frac{|V|-1}{2}$ edges, hence $\nu_2(G)\leq |V|-1$. Now, if we remove any edge from $C$, then the resulting Hamiltonian path will be a 2-edge-colorable subgraph with $|V|-1$ edges. Hence it will be a maximum 2-edge-colorable subgraph in $G$. The proof is complete.
\end{proof}

\begin{remark}
	Let us note that the proof of Ore's theorem represents a polynomial time algorithm which actually finds the Hamiltonian cycle. Thus, in the second case of the previous proof, the algorithm will run in polynomial time.
\end{remark}

Observe that in any graph $G$, we have the following relationship among vertex, edge connectivity and minimum degree:
\[\kappa(G)\leq \kappa'(G)\leq \delta(G).\]
Since, the maximum 2-edge-colorable subgraph problem is paraNP-hard with respect to $\delta$, Lemma \ref{lem:Redparam} implies that the problem is paraNP-hard with respect to $\kappa(G)$ and $\kappa'(G)$. Moreover, since in any graph
\[ |V|- \delta(G)\leq |V|-\kappa'(G)\leq |V|-\kappa(G),\]
Proposition \ref{prop:Vminusdelta} and Lemma \ref{lem:Redparam} imply that the problem is FPT with respect to $|V|-\kappa(G)$ and $|V|-\kappa'(G)$.

In the following, we show that the maximum 2-edge-colorable subgraph problem is fixed-parameter tractable with respect to the \textit{treewidth} of the graph. We first recall some basic definitions. Informally, a \textit{tree decomposition} of a graph $G$ is a way of representing $G$ as a tree-like structure. 

\begin{definition}[\cite{fomin:2010}]\label{def:path_decomposition2}
A tree decomposition of a graph $G=(V,E)$ is a pair $(T, \{X_i\}_{i\in V(T)})$, where $T$ is a tree, and for each vertex $i$ of $T$, $X_i$ is a subset of vertices of $G$ called \textit{bags}, such that: (i) for every $u\in V$ there exists $i\in V(T)$ with $u\in X_i$; (ii) for every $(u,v)\in E$, there exists $i\in V(T)$ with $u,v\in X_i$; (iii)  for every $u\in V$, the set $V_u(T)=\{i\in V(T)|u \in X_i\}$ forms a connected subgraph (subtree) of $T$.
\end{definition}

%


The \textit{width} of a tree decomposition $(T, \{X_i\}_{i\in V(T)})$ equals  $\max_{i\in V(T)}|X_i|-1$, and the \textit{treewidth} of a graph $G$, is the minimum width of a tree decomposition of $G$. An example of tree decomposition is reported in Figure~\ref{fig:treedecomposition} for the graph depicted in Figure~\ref{fig:treewidth}. 
\begin{figure}[t]
\centering
\begin{minipage}[b]{0.43\textwidth}
\includegraphics[width=\textwidth]{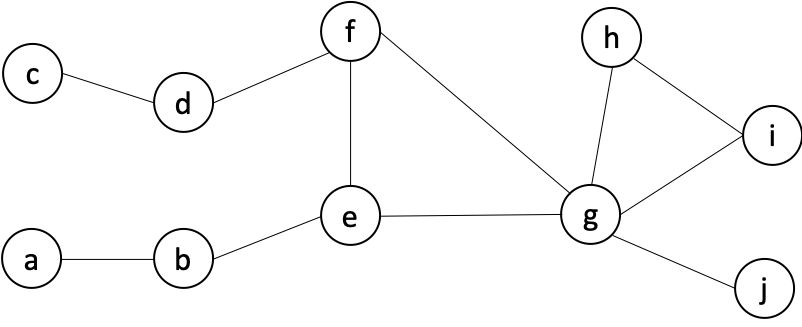}\caption{An example of a graph having treewidth 2.}\label{fig:treewidth}
\end{minipage}
\hfill
\begin{minipage}[b]{0.32\textwidth}
\includegraphics[width=\textwidth]{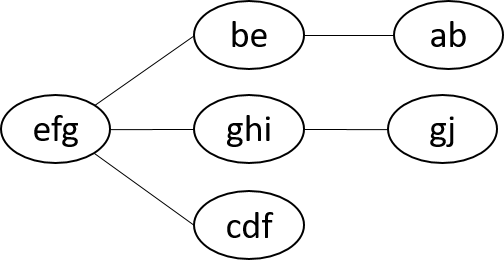}
\caption{A tree decomposition of width 2 for the graph shown in Figure~\ref{fig:treewidth}.}\label{fig:treedecomposition}
\end{minipage}
\end{figure}

\begin{theorem}
\label{thm:TreeWidth} Given a tree decomposition of width $h$, the maximum 2-edge-colorable subgraph problem is FPT with respect to $h$. In other words, the problem is FPT with respect to the treewidth.
\end{theorem}
\begin{proof} Our main tool for proving this theorem is Courcelle's theorem (see section 7.4 of \cite{FPTbook}), which states that any graph property that can be defined on monadic second-order logic can be tested in time $f(h)\cdot poly(size)$. Here $f$ is some computable function and $size$ is the length of the input. Thus, in order to complete the proof, we need to express the maximum 2-edge-colorable subgraph problem in the monadic second-order logic.

First, we write that the subgraph has to be partitioned into two matchings:\\

$\max|S|:S\subseteq E,\text{ } \exists_{A, B\subseteq S} \text{ }  Partition(A,B)\text{ } \land \text{ }  Match(A)\text{ } \land \text{ } Match(B).$\\


Now, let us present the formula for $Partition(A,B)$ that states that $S$ is partitioned into $A$ and $B$:
\[Partition(A,B)=\forall_{e\in S}\text{ } (e\in A \text{ } \land \text{ } e\notin B)\text{ } \lor \text{ }(e\notin A\text{ }\land \text{ } e\in B).\]
Finally, let us write a formula that states that a subset is a matching:
\[Match(C)=\forall_{e_1, e_2\in C} (adj(e_1, e_2)\rightarrow e_1=e_2),\]
where
\[adj(e_1, e_2)=\exists_{u,v, w\in V}\text{ } e_1=\{u,v\}\text{ } \land \text{ } e_2=\{w,v\}.\]
The proof is complete.
\end{proof}

The parameterization with respect to the treewidth implies the folllowing:

\begin{corollary}
\label{cor:FPTnu1} The maximum 2-edge-colorable subgraph problem is FPT with respect to $\nu(G)$.
\end{corollary} This just follows from Theorem \ref{thm:TreeWidth} and Lemma \ref{lem:Redparam} as it is demonstrated in \cite{Sasak:2010}.

Let $\alpha'(G)$ be the smallest number of edges of $G$ such that any vertex of $G$ is incident to at least one of these edges. By the classical Gallai theorem \cite{west:1996}, we have that if the graph $G$ has no isolated vertices, then
\[\nu(G)+\alpha'(G)=|V|.\]
Since $\nu(G)\leq \frac{|V|}{2}$, we have
\[\nu(G)\leq \frac{|V|}{2} \leq \alpha'(G).\]
Thus, Corollary \ref{cor:FPTnu1} and Lemma \ref{lem:Redparam} imply that the maximum 2-edge-colorable subgraph problem is FPT with respect to $\alpha'(G)$. Observe that isolated vertices play no role in the maximum $k$-edge-colorable subgraph problem, thus we can assume that the input graph contains none of them.

Also, observe that the parameterization with respect to $\alpha'(G)$ can be interpreted as parameterization with respect to $|V|-\nu(G)$. One may wonder, whether we can strengthen this result, by showing that the maximum 2-edge-colorable subgraph problem is FPT with respect to $|V|-2\cdot \nu(G)$? The answer to this question is negative unless $P=NP$. If a cubic graph $G$ is 3-edge-colorable, then it must be bridgeless. Thus, by Holyer's result the maximum 2-edge-colorable subgraph problem is $NP$-hard for bridgeless cubic graphs. By the classical Petersen theorem \cite{west:1996}, bridgeless cubic graphs have a perfect matching. Thus, in this class we have $|V|-2\nu(G)=0$. Hence, the maximum 2-edge-colorable subgraph problem is paraNP-hard with respect to $|V|-2\nu(G)$.

One can consider the decision version of the maximum 2-edge-colorable subgraph problem, where for a given graph $G$ and an integer $t$, one needs to check whether $\nu_2(G)\geq t$. It turns out that this problem is FPT with respect to $t$. In order to see this, just observe that if in the input graph $G$ $\nu(G)\geq t$, then clearly $\nu_2(G)\geq \nu(G)\geq t$, hence the instance is a ``yes" instance. On the other hand, if $\nu(G)\leq t$, then the FPT algorithm with respect to $\nu(G)$ (Corollary \ref{cor:FPTnu1}) will in fact be an FPT algorithm with respect to $t$ (Lemma \ref{lem:Redparam}).


Below we will parameterize the maximum 2-edge-colorable subgraph problem with respect to the dimension of the cycle space of a graph. Before we start, we recall some concepts. Let $B=\{0,1\}$ be the binary field. If $K$ is a subset of edges of a graph $G$, then we can consider an $|E|$-dimensional vector $x_K$ whose coordinates are zero and one, and for any edge $e$ of $G$, we have $e\in E'$, if and only if $x_K(e)=1$. Here $x_K(e)$ denotes the coordinate of $x_K$ corresponding to $e$. The vector $x_K$ is usually called the characteristic vector of $K$. Observe that the characteristic vectors of all subsets of $E$ form a $|E|$-dimensional linear space over $B$. Now, let us consider the cycle space $C(G)$ of $G$, which is defined as the linear hull of all characteristic vectors that correspond to simple cycles of $G$. Clearly, the cycle space is a linear subspace of all characteristic vectors. A classical result in the area states that if $G$ is any graph with $k$ components, then the dimension of $C(G)$ is given by the following formula:
\[dim(C(G))=|E|-|V|+k.\]
An alternative way of looking at $dim(C(G))$ is the following: a subset $F$ of edges of a graph $G$ is called a feedback edge-set, if $G-F$ contains no cycles. In other words, any cycle of $G$ contains an edge from $F$. It turns out that $dim(C(G))$ represents the size of a smallest feedback edge-set of $G$. Moreover, there is a polynomial time algorithm that finds such a subset of edges for any input graph $G$.

For our parameterization, we will require the following lemma:

\begin{lemma}
\label{lem:ForestConstraints} Let $G$ be a forest, and let $S$ be a set of vertices of $G$. Assume that $S=S_1\cup S_2$ is a partition of $S$. Then there is a linear time algorithm that finds a largest 2-edge-colorable subgraph such that the vertices of $S_j$ are not incident to edges with color $j$ ($j=1,2$). 
\end{lemma}

\begin{proof} Clearly we can assume that $G$ is a tree, otherwise we can solve the problem on each component of the forest, then take the union of these solutions. Moreover, the constraints on the vertices in $S$ can be seen as: only color $\mathit{2}$ can be used on the edges incident to a vertex in $S_1$; and only color $\mathit{1}$ can be used on the edges incident to a vertex in $S_2$. In the following, we say that a color is \textit{incident} to a vertex, if it is used at least on one of the edges incident to the vertex. In order to describe the algorithm, we add the extra dummy color $\mathit{0}$ that means 'not colored', so the set of available colors becomes $\{\mathit{0},\mathit{1},\mathit{2}\}$, where $\mathit{1}$, and $\mathit{2}$ are called true colors.

Let $G(V,E)$ be a tree with $n=|V|$ vertices and $n-1=|E|$ edges. Assume that $W\colon V\to 2^{\{\mathit{0},\mathit{1},\mathit{2}\}}$ is an allocation of the available colors for each vertex in $V$. Let $p\colon \{\mathit{0},\mathit{1},\mathit{2}\}\to \{0,1\}$ be a profit function that is equal to 0, if the input is the dummy color $\mathit{0}$, and is 1, otherwise. Let $r$ be the root of $G$, and let $G(u)$ be the subgraph induced by $u$ and all the descendants of $u$ in $G$.

Now, we describe a \textit{dynamic programming algorithm} that finds a largest 2-edge-colorable subgraph on $G=(V,E)$, with the additional constraints where the colors that can be \textit{incident} to each vertex  $u\in V$ are the ones in $W(u)$. We call this problem $\mathcal{P}$. Clearly, we can write the original problem of the lemma as $\mathcal{P}$ with a specific color allocation, where $W(u)=\{\mathit{0},\mathit{2}\}$ for each vertex $u\in S_1$; $W(u)=\{\mathit{0},\mathit{1}\}$ for each vertex $u\in S_2$; and $W(u)=\{\mathit{0},\mathit{1},\mathit{2}\}$ for each vertex $u\in V\setminus S$.	

In the algorithm, we compute $f(u,A)$ for each vertex $u\in V$, and for every $A\subseteq W(u)$, which is equal to the optimum value of $\mathcal{P}$ restricted to the subgraph $G(u)$, and with the additional constraint where the colors incident to $u$ are those in $A$. If there is no solution, we set $f(u,A)=-\infty$. In particular, the algorithm starts from the leaves and goes up to the {root} $r$, and the optimal value of $\mathcal{P}$ is the maximum $f(r,A)$, for every $A\subseteq W(r)$.

If $u$ is a leaf, and $A$ a specific color subset of $W(u)$, one of the two conditions holds: $f(u,A)=0$, if $A=\emptyset$;  $f(u,A)=-\infty$, otherwise. In fact, since there are no edges in $G(u)$ when $u$ is a leaf, there cannot be any color incident to $u$ in $G(u)$.

For each internal vertex $u$, we suppose $A\neq \emptyset$, as it is always possible to use the dummy color $\mathit{0}$. If $u$ is an internal vertex with $t$ sons $\{v_1,\ldots,v_t\}$, we compute $f(u,A)$ for any subset $A\subseteq W(u)$ by using the values $f(v_i,A)$ for every $u$'s son $v_i$. For every $i\in\{1,\ldots,t\}$, denote by $V_i$ the set containing the vertices in the subgraph $G(u)$ minus the vertices in each subgraph $G(v_j)$, with $j\in \{i+1,\ldots,t\}$. Let $G(V_i)$ be the subgraph induced by the vertices in $V_i$. For every $i\in\{1,\ldots,t\}$ we compute $h(u,V_i,A)$, which equals the maximum value of $\mathcal{P}$ restricted to the subgraph $G(V_i)$, and with the following additional constraint, where the colors incident to $u$ are those in $A$. If there is no solution, we set $h(u,V_i,A)=-\infty$. Notice that  $h(u,V_t,A)$ is equivalent to $f(u,A)$. Now, we see how to compute $h(u,V_i,A)$ for every $i\in\{1,\ldots,t\}$, recursively.

If $i=1$, there is only one edge incident to $u$ in $G(v_1)$, i,.e., $(u,v_1)$. So, we can set $A=\{c(u,v_1)\}$, where $A$ contains only the color assigned to $(u,v_1)$. We compute $h(u,V_1,A)$ solving the following problem.

\begin{equation} 
\begin{aligned}  \label{eq:tree_S_g1}
& {\text{max}}
& &  f(v_1,C) + p(c(u,v_1)) \\
& \text{s.t.} & &C\cap  \{c(u,v_1)\} \subseteq \{\mathit{0}\}\\  
& & & C\subseteq W(v_1)
\end{aligned}
\end{equation}
%


In fact, for a specific color $c(u,v_1)$, we get the best value $f(v_1,C)$ for every $C\in W(v_1)$ that is compatible with $c(u,v_1)$. The compatibility is guaranteed by the constraint that does not allow to choose a subset $C$ (the problem's variable) that contains $c(u,v_1)$ if it is a true color.

If $i\geq 2$, we calculate $h(u,V_i,A)$ by using the values $h(u,V_{i-1},B)$. We solve the following maximisation problem for every non empty subset $B\subseteq W(u)$ and any color $c(u,v_i)\in W(u)$ for the edge $(u,v_i)$, such that $A=B\cup \{c(u,v_i)\}$, and $B\cap \{c(u,v_i)\}\subseteq \{\mathit{0}\}$.

\begin{equation} 
\begin{aligned}  \label{eq:tree_S_g}
& {\text{max}}
& &  h(u,V_{i-1},B) + f(v_i,C) + p(c(u,v_i)) \\
& \text{s.t.} & &C\cap  \{c(u,v_i)\} \subseteq \{\mathit{0}\}\\  
& & &C\subseteq W(v_i)
\end{aligned}
\end{equation}

The idea is that, for every $B\subseteq W(u)$, and for every $c(u,v_i)$ that are compatible, we search the subset $C\subseteq W(v_i)$ (the problem's variable), compatible with $c(u,v_i)$, which maximise the number of edges with true colors in $G(V_i)$, that is the objective function. In the objective function, $f(v_i,C) + p(c(u,v_i))$ refers to the subgraph $G(v_i)\cup (u,v_i)$, while $h(u,V_{i-1},B)$ is the value already computed. 

The time complexity to compute $h(u,V_i,A)$ for all the $t$ sons of an internal vertex $u$ is $O(3\cdot 2^6 t)=O(t)$. In fact, for $t=1$ we solve Problem~\ref{eq:tree_S_g1} for every $c(u,v_i)\in W(u)$, so at most $3$ times; at each of the $t$ steps, with $t\geq 2$, we solve Problem~\ref{eq:tree_S_g} at most for every subset $B\subseteq W(u)$, and for every color $c(u,v_i)\in W(u)$. Since there are less that $2^3$ non empty subsets $B$, at most $3$ possible colors for $c(u,v_i)$, and at most $2^3$ subsets $C$ (the problems' variable), the time complexity for computing $h(u,V_i,A)$ for every $u$'s sons, is $O(3\cdot 2^6 t)$. Then, this is the time complexity to compute $f(u,A)$ for an internal vertex $u$ with $t$ sons. 

In conclusion, starting from the leaves, we can compute $f(u,A)$ for every internal node $u$, from the lowest level of the tree until we reach $r$. Since we need $O(3\cdot 2^6 t_u)=O(t_u)$ time for each internal node $u$ with $t_u$ sons, the total running-time will be $\sum_{u\in V}O(t_u)$ which is $O(|V|)$, as the number of edges in a tree is $|V|-1$. The proof is complete.
\end{proof}

We are ready to obtain the next result.

\begin{theorem}
\label{thm:FPTdimCycSpace} The maximum 2-edge-colorable subgraph problem is FPT with respect to the dimension of the cycle space.
\end{theorem}

\begin{proof} Let $k=dim(C(G))$. In polynomial time, we can find $k$ edges whose removal leave a forest. Let $F$ be this set of $k$ edges, and let $T=G-F$. Any maximum 2-edge-colorable subgraph of $G$ colors some subset of edges of $F$. Thus, we can guess this subset. The number of choices is $3^k$ (each edge of $F$ is either of color 1 or 2, or 0 meaning that it is uncolored). Now, consider any of these guesses. If it contains at least three edges adjacent to the same vertex, or two edges of the same color incident to the same vertex, then we do not consider it. If it contains two edges $e$ and $f$ incident to the same vertex $z$ such that edges have different color, we remove $z$ and forbid the corresponding color on the other end-point of $e$ and $f$ in $T$. If an edge is not adjacent to any other edge in the guess, we simply remove it and forbid its color in its end-points on $T$. Having done this, we get an instance of the forest problem with constraints on vertices. By Lemma \ref{lem:ForestConstraints}, we can find a largest 2-edge-colorable subgraph respecting the constraints in polynomial time. Thus, we can compare the sizes of all these 2-edge-colorable subgraphs and get a maximum 2-edge-colorable subgraph of $G$ in polynomial time. Thus, the total running-time of our algorithm is $3^{k}\cdot poly(size)$. The proof is complete.
\end{proof}

\begin{remark}
\label{rem:BipNumberTreewidth} Theorem \ref{thm:FPTdimCycSpace} can be deduced as a consequence of Theorem \ref{thm:TreeWidth} and Lemma \ref{lem:Redparam}, since by Lemma 1 of \cite{SongYu:2015}, the size of the smallest feedback edge-set is an upper bound for the size of the smallest feedback set, which in its turn is an upper bound for the treewidth. 
\end{remark} Despite this ``negative" remark, we believe that our proof is interesting, as it relies on Lemma \ref{lem:ForestConstraints} which can be useful in other situations. Also note that our proof allows us to obtain an explicit expression for the running-time of the algorithm. As it is stated in the end of Section 7.4.2 of \cite{FPTbook}, obtaining the exact expression for the running-time of algorithms arising from Courcelle's theorem could be a non-trivial task. The reader is invited to take a look at the end of Section 7.4.2 of \cite{FPTbook} for further details on this.

The strategy of the proof of Theorem \ref{thm:FPTdimCycSpace} implies the following corollary:
\begin{corollary}
\label{cor:AtmostLogEdges} Let $G=(V,E)$ be a connected graph with $|E|\leq |V|+\log |V|$ edges. Then the maximum 2-edge-colorable subgraph problem can be solved in polynomial time for this type of graphs.
\end{corollary}

\begin{proof} The proof is the same. Start with any spanning tree $T$ of $G$. Observe that the number of edges of $G$ outside $T$ is at most $\log |V|$. Guess all possible assignments of 2-colors to these edges. Since their number is at most $\log |V|$, we have that the total number of guesses in polynomial in $|V|$. For each of the guesses, via Lemma \ref{lem:ForestConstraints}, we find a largest 2-edge-colorable subgraph respecting the constraints arising from the guesses in polynomial time. Thus, we can compare the sizes of all these 2-edge-colorable subgraphs and get a maximum 2-edge-colorable subgraph of $G$ in polynomial time. The proof is complete.
\end{proof}

Using this corollary, one can show that our problem is FPT with respect to $MaxLeaf(G)$. Recall that for a connected graph $G$, $MaxLeaf(G)$ is defined as the maximum number of leaves in a spanning tree of $G$. In order to derive this result, we will use the following

\begin{theorem}(\cite{Ding})
\label{thm:DingMaxLeafThm}
Let $G$ be a simple connected graph with $|E|\geq |V|+ \frac{t(t-1)}{2}$ edges and $|V|\neq t+2$. Then $MaxLeaf(G)>t$ and the bound is best possible.
\end{theorem}

We are ready to prove:
\begin{proposition}
\label{prop:MaxLeafFPT} The maximum 2-edge-colorable subgraph problem is FPT with respect to $MaxLeaf(G)$.
\end{proposition}

\begin{proof} Clearly, we can assume that the input graph $G$ is connected. If $|E|\leq |V|+\log |V|$, then Corollary \ref{cor:AtmostLogEdges} implies that we can find a maximum 2-edge-colorable subgraph in polynomial time. Thus, without loss of generality, we can assume that $|E|> |V|+\log |V|$. Observe that if we choose $t=\lfloor\sqrt{2\log|V|}\rfloor$, then 
\[\log |V|\geq \frac{t^2}{2},\]
hence
\[|E|>|V|+\log |V|\geq |V|+\frac{t^2}{2}>|V|+\frac{t(t-1)}{2}.\]
Since $|V|\neq t+2=\lfloor \sqrt{2\log|V|}\rfloor +2$ (this is true for sufficiently large $|V|$, we can solve small instances with the brute force algorithm directly), Theorem \ref{thm:DingMaxLeafThm} implies that 
\[ MaxLeaf(G)>t=\lfloor\sqrt{2\log|V|}\rfloor>\sqrt{2\log|V|}-1,\]
or
\[|V|< 2^{\frac{(MaxLeaf(G)+1)^2}{2}}.\]
In other words, $|V|$ is bounded in terms of $MaxLeaf(G)$. Thus the trivial algorithm will solve the problem in FPT time with respect to $MaxLeaf(G)$. The proof is complete.
\end{proof}

One may wonder, whether our problem is FPT with respect to the complementary parameter $|V|-MaxLeaf(G)$? Observe that $|V|-MaxLeaf(G)=1$ if and only if $MaxLeaf(G)=|V|-1$. The latter condition is equivalent to the statement that the graph under consideration contains a spanning star. However, the latter condition is the same as having $rad(G)=1$. Thus, combined with Theorem \ref{thm:radius}, we get:
\begin{proposition}
\label{prop:|V|minusMaxLeaf} The maximum 2-edge-colorable subgraph problem is paraNP-hard with respect to $|V|-MaxLeaf(G)$.
\end{proposition}


\section{Dynamic programming algorithms for branchwidth, treewidth, and cliquewidth}\label{sec:dynamic}

This section is dedicated to two dynamic programming algorithms for three well known parameters, that is branchwidth, treewidth, and cliquewidth.

\subsection{An FPT dynamic programming algorithm for graphs with bounded branchwidth and treewidth }

We describe here an FPT algorithm for the maximum 2-edge-colorable subgraph problem on graphs with bounded \textit{branchwidth} and \textit{treewidth}. Even if Theorem~\ref{thm:TreeWidth} already states that the maximum 2-edge-colorable subgraph is FPT when the graph has bounded treewidth, it is worth to provide an FPT algorithm with a good execution time bound. In fact, Theorem~\ref{thm:TreeWidth} uses the Courcelle's result, which in general has a very high bound of the function $f(\theta)$.  We first recall some basic definitions that we use in Theorem~\ref{th:BranchWidth}. 

\begin{definition}[\cite{FPTbook}]\label{def:branch_decomposition2}
A branch decomposition of a graph $G=(V,E)$ is a pair $(T, \phi)$, where $T$ is a tree with $|E|$ leaves whose all internal nodes have degree three, and  $\phi$ is a  bijection from $E$ to the leaves of $T$. Each edge $(i,j)$ of $T$ divides the tree into two components, so it divides $E$ into two parts $X$ and $E\setminus X$. We denote by $\delta(X)$ (\textit{border}) the set of those vertices of $G$ that are both incident to an edge of $X$ and to an edge of $E\setminus X$. The \textit{width} of a decomposition $(T,\phi)$ is the maximum value on the edges of $T$. The \emph{branchwidth} of $G$, denoted by $bw(G)$, is the minimum width over all branch decompositions of $G$. We put $bw(G)=0$ if $|E|=1$, and  $bw(G)=0$ if $|E|=0$.
\end{definition}

Let $G(V,E)$ be a graph with a branch decomposition $(T,\phi)$ of width $h$, where  \emph{root} is the root of $T$. Now, we describe a dynamic programming algorithm that solves \textit{maximum 2-edge-colorable subgraph} for $G$, which exploits the structure and the properties of branch decomposition $(T,\phi)$. We will call nodes the vertices of $T$, in order to avoid confusion between the vertices of the graph and the ones of $T$.

Similarly to Lemma~\ref{lem:ForestConstraints}, we assume that $\{\mathit{0},\mathit{1},\mathit{2}\}$ is the set of colors where $\mathit{1}$, and $\mathit{2}$ are true colors, while $\mathit{0}$ is dummy and means 'not colored'. Let $p\colon \{\mathit{0},\mathit{1},\mathit{2}\}\to\{0,1\}$ be a function, which is equal to 1 if and only if the input is a true color, 0 otherwise (i.e., $p(\mathit{0})=0$, $p(\mathit{1})=1$, and $p(\mathit{2})=1$). If $c\colon E\to \{\mathit{0},\mathit{1},\mathit{2}\}$ is a function, then $c((u,v))$ is called the color assigned to the edge $(u,v)$. In order to avoid cluttered notation, in the following we will write $c(u,v)$ instead of $c((u,v))$. 

Denote by $T(i)$ the subtree induced by the node $i$ of $T$ and all its descendants. Denote also by $E_i$ the leaves of $T(i)$, by $G(i)$ the subgraph of $G$ induced by the edges in $E_i$, and by $\delta(E_i)$ (\textit{border}) the set of those nodes of $G$ that are both incident to an edge of $E_i$ and to an edge of $E\setminus E_i$. In the core of the dynamic programming algorithm, we compute the optimal value of a constrained version of the problem restricted to the subgraph $G(i)$. The constraint is given by the colors incident on each node in the border $\delta(E_i)$. In particular, we compute $f(i,\mathcal{A})$, which is the optimum value of \textit{maximum 2-edge-colorable subgraph} on $G(i)$, where $\mathcal{A}$ is a collection of subsets of colors defined as

\begin{itemize}
\item The colors \textit{incident} to the vertex $u$ are those in $A(u)$, for every $u\in X_i$.
\end{itemize}

Here, as usual, a color is \textit{incident} to a vertex $u$ if it is used at least in one edge incident to $v$. Please note that only the dummy color $\mathit{0}$ can be incident more than once to a vertex, because it means that an edge is not colored. Please also note that the branchwidth of $G$ is $h$, so the maximum number of subsets in $\mathcal{A}$ is exactly $h$. 

If $i$ is a leaf corresponding to an edge $(u,v)\in E$, then there are three different cases. In the first one both $u$ and $v$ are not connected to some other nodes in $V\setminus \{u,v\}$, then $\mathcal{A}=\emptyset$, because the border is empty, and $f(i,\mathcal{A})=1$. In fact, the edge $\{u,v\}$ is isolated, so, we color it with color 1 or 2 indifferently.

In the second case only one of the two nodes $u$ and $v$ is connected with the rest of the graph. We can suppose w.l.o.g. that $u$ is the connected node.  Then, $\mathcal{A}=\{A(u)\}$ and

\begin{itemize}
\item $f(i,\mathcal{A})=0$ \quad if $A(u)=\{\mathit{0}\}$;
\item $f(i,\mathcal{A})=1$ \quad if $A(u)=\{\mathit{1}\}$;
\item $f(i,\mathcal{A})=1$ \quad if $A(u)=\{\mathit{2}\}$;
\item $f(i,\mathcal{A})=-\inf$ \quad otherwise.
\end{itemize}

In fact, the number of edge colored is 1 iff there is a real color incident to $u$. Analogously, in the third case both the nodes $u$ and $v$ are connected to the rest of the graph, so, $\mathcal{A}=\{A(u),A(v)\}$ and

\begin{itemize}
\item $f(i,\mathcal{A})=0$ \quad if $A(u)=A(v)=\{\mathit{0}\}$;
\item $f(i,\mathcal{A})=1$ \quad if $A(u)=A(v)=\{\mathit{1}\}$ or $A(u)=A(v)=\{\mathit{2}\}$;
\item $f(i,\mathcal{A})=-\inf$ \quad otherwise.
\end{itemize}

In fact, to get 1 as profit, the edge $\{u,v\}$ is colored, and there is only one color incident both in $u$ and $v$.

We analise now the case of an internal node $i$ of $T$ with children $j$ and $k$. We can compute $f(i,\mathcal{A})$ by using the already computed values $f(j,\mathcal{B})$ and $f(k,\mathcal{C})$. Please note that $G(j)$ and $G(k)$ can share some common vertices, that is $V_j\cap V_k\neq \emptyset$. If it is the case, the solutions on $G(j)$ and $G(k)$ must have the same colors incident on these nodes, because we want to merge them. The following constrained maximisation problem solves $f(i,\mathcal{A})$ for a particular collection $\mathcal{A}$.

\begin{equation} \small
\begin{aligned} \label{eq:Branch2}
& {\text{max}}
& & f(j,\mathcal{B})+f(k,\mathcal{C})\\
& \text{s.t.} & & A(u)=B(u) \quad \forall u\in ((\delta(E_j)\setminus \delta(E_k))\cap \delta(E_i))\\
& & & A(u)=C(u) \quad \forall u\in ((\delta(E_k)\setminus \delta(E_j))\cap \delta(E_i))\\
& & & A(u)=B(u)\cup C(u) \quad \forall u\in ((\delta(E_j)\cap \delta(E_k))\cap \delta(E_i))\\
& & & B(u)\cap C(u)\subseteq \{\mathit{0}\} \quad \forall u\in (\delta(E_j)\cap \delta(E_k))
\end{aligned}
\end{equation}

In fact, the border of $G(i)$ is a subset of the union of the borders of $G(j)$ and $G(k)$, that is $\delta(E_i)\subseteq \delta(E_j)\cup \delta(E_k)$. Then, in order to merge the solutions founded for $G(j)$ and $G(k)$, we must guarantee that: for each node $u$ in $\delta(E_i)$ that is also in $\delta(E_j)$ but not in $\delta(E_k)$, the set of colors $A(u)$ equals $B(u)$ (first constraint); for each node $u$ in $\delta(E_i)$ that is also in $\delta(E_k)$ but not in $\delta(E_j)$, the set of colors $A(u)$ equals $C(u)$ (second constraint); for each node $u$ in $\delta(E_i)$ that is also in $\delta(E_k)$ and in $\delta(E_j)$, the set of colors $A(u)$ equals $C(u)\cup B(u)$ (third constraint); the set of colors incident in each node belonging both in $\delta(E_j)$  and $\delta(E_k)$ must be compatible (fourth constraint).

We analise now the time complexity of the algorithm. 

If $i$ is a leaf of $T$, since the graph $G(i)$ is only made by an edge $\{u,v\}$, there are at most only three ways for the collections $\mathcal{A}=\{A(u),A(v)\}$, which are $\{\{\mathit{0}\},\{\mathit{0}\}\}$, $\{\{\mathit{1}\},\{\mathit{1}\}\}$, and $\{\{\mathit{2}\},\{\mathit{2}\}\}$.

If $i$ is an internal node, the time complexity is given by all the possible combination for the two collections $\mathcal{B}$ and $\mathcal{C}$, which have cardinality at most $h$. The number of possible sets $A(u)$, for a given vertex $u$, are at most 6, that is $\{\mathit{0}\}$, $\{\mathit{1}\}$, $\{\mathit{0},\mathit{1}\}$, $\{\mathit{0},\mathit{2}\}$, $\{\mathit{1},\mathit{2}\}$, $\{\mathit{0},\mathit{1},\mathit{2}\}$. This means that the time complexity at each internal node is $O(6^{2h})$

In conclusion, since each internal node has degree three (two sons), and in a perfect binary tree there are $2|V|-1$ nodes, then $T$ has at most $2|V|-1$ nodes. This means that the time complexity of the dynamic programming algorithm is $O((2|V|-1)6^{2h})$.

We can then state the following:

\begin{theorem}\label{th:BranchWidth}
Given an instance of \textit{maximum 2-edge-colorable subgraph problem} with a graph $G$, and a branch decomposition $T$ of width $h$ for $G$, the \textit{maximum 2-edge-colorable subgraph problem} is solvable in $O(6^{2h}(2|V|-1))$ time.
\end{theorem}

It can be proved that branchwidth is in some sense equivalent to  treewidth, expressed in the formalism of branch decomposition. In fact, the following inequalities \cite{fomin:2015}, valid for every graph $G$ with branchwidth $h>1$, set the relations between these two parameters. 

\begin{equation} \label{bw_tw}
\text{bw}(G)\leq \text{tw}(G)+1\leq \frac{3}{2}\text{bw}(G)
\end{equation}

This allow us to write the following corollary.

\begin{theorem}\label{thm:TreeWidth2}
Given an instance of \textit{maximum 2-edge-colorable subgraph problem} with a graph $G$ with treewidth $h$, the \textit{maximum 2-edge-colorable subgraph problem} is solvable in $O(6^{\frac{4}{3}(h+1)}(2|V|-1))$ time.
\end{theorem}
\begin{proof}
By using the inequalities \ref{bw_tw}, it turns out that the branchwidth $h'$ of $G$ is greater or equal to $\frac{2}{3}(h+1)$. By using Theorem\ref{th:BranchWidth}, we can optimally solve the problem in $O(6^{2h'}(2|V|-1))$, that is in $O(6^{\frac{4}{3}(h+1)}(2|V|-1))$.

\end{proof}

\subsection{A polynomial time dynamic programming algorithm for graphs with bounded cliquewidth}

Here we show that our problem is polynomial time solvable in the class of graphs having bounded cliquewidth. We start by formally defining what the \textit{cliquewidth} \cite{FPTbook} of a graph is. Subsequently, we present our algorithm that optimally solves the \textit{maximum 2-edge-colorable subgraph problem} on graphs with bounded cliquewidth, which runs in polynomial time of the size of the instance.

We briefly recall some basic definitions before describing the algorithm. We are given an undirected graph $G=(V,E)$ with $n$ vertices and $m$ edges. The cliquewidth of $G$ is the minimum number $h$ of labels needed to build $G$ by using the following operations.

\begin{itemize}
\item i) \textit{Creation} of a new vertex $u$ with label $i$, denoted by $u_i$;
\item ii) \textit{Disjoint union} of two labelled graphs $G_1$ and $G_2$, denoted by $G_1 \oplus G_2$;
\item iii) \textit{Joining} by an edge every vertex labelled $i$ to every vertex labelled $j$;
\item iv) \textit{Renaming} all vertices with label $i$ to label $j$, denoted $\rho_{i\rightarrow j}(G)$.
\end{itemize}


We denote with $\varphi \colon V \rightarrow [h]$ a function that maps each vertex $u\in V$ to a specific label $i\in [h]$. We also denote with $\mathcal{X}$ a $k$-expression made of the four operations above, which is a clique decomposition of $G$. We check now that, given a decomposition $\mathcal{X}$ of $G$, we can compute an optimal solution of the \textit{maximum 2-edge-colorable subgraph problem} in polynomial time. In particular, we describe a dynamic programming algorithm, which exploits the fact that $G$ can be built by using the four operations above. 


We denote with $f(G,A)$ the optimum value  for the \textit{maximum 2-edge-colorable subgraph problem} on the labelled graph $G$, with some additional constraints. The set of vertices $V$ is divided into $h$ disjoint subsets $V_i\subseteq V$,  following the labels. In other words, all the nodes in each $V_i$ are labelled with label $i$. The set $A$ contains $8(h-1)+1$ non negative integers $a_{i}^0$, $a_{ij^-}^1$, $a_{ij}^1$, $a_{ij^-}^2$ $a_{ij}^2$, $a_{ij^-}^{12}$,  $a_{i1j}^{12}$, $a_{i2j}^{12}$, $a_{ij}^{12}$ for every label $i$ in $[k]$. Please note that $i$ and $j$ are two different labels in $[k]$. We denote with $A_i$ the subset of $A$ related to label $i$, that is $A_i=\{a_i^0, a_{ij^-}^1, a_{ij}^1, a_{ij^-}^2, a_{ij}^2, a_{ij^-}^{12}, a_{i1j}^{12}, a_{i2j}^{12}, a_{ij}^{12} \}$ for every $j\in ([k]\setminus i)$. Similarly, we denote with $A_{ij}$ the subset of $A$ related to two disjoint labels $i,j$, that is $A_{ij}=\{a_i^0, a_{ij^-}^1, a_{ij}^1, a_{ij^-}^2, a_{ij}^2, a_{ij^-}^{12}, a_{i1j}^{12}, a_{i2j}^{12}, a_{ij}^{12}\}$. Please note that, for a fixed $i\in [k]$, $a_i^0$ is contained in every $A_{ij}$ for every $j\in ([k]\setminus i)$.

The constraints related to the values $f(G,A)$ are the following ones.

\begin{itemize}
\item There are $a_{i}^0$ vertices in $V_i$ with zero  colors incident to them;
\item there are $a_{ij^-}^1$ vertices in $V_i$, each with only color 1 incident to it via an edge $(u,v)$ with $u\in V_i$ and $v\notin V_j$ ($v$ can belong to $V_i$);
\item there are $a_{ij}^1$ vertices in $V_i$, each with only color 1 incident to it via an edge $(u,v)$ with $u\in V_i$ and $v\in V_j$;
\item there are $a_{ij^-}^2$ vertices in $V_i$, each with only color 2 incident to it via an edge $(u,v)$ with $u\in V_i$ and $v\notin V_j$ ($v$ can belong to $V_i$);
\item there are $a_{ij}^2$ vertices in $V_i$, each with only color 2 incident to it via an edge $(u,v)$ with $u\in V_i$ and $v\in V_j$;
\item there are $a_{ij^-}^{12}$ vertices in $V_i$ with both colors 1 and 2 incident to them, but all the edges from this subset of $V_i$ to any vertex in $V_j$ are not colored (the two colored edges can completely belong to $V_i$); 
\item there are $a_{i1j}^{12}$ vertices in $V_i$ each with both colors 1 and 2 incident to them. In particular, for each $u$ of these vertices, the color 1 is on an edge $(u,v)$ with $v\in V_j$; while the color 2 is on an edge $(u,z)$ with $z\notin V_j$ ($z$ can belong to $V_i$);
\item there are $a_{i2j}^{12}$ vertices in $V_i$ each with both colors 1 and 2 incident to them. In particular, for each $u$ of these vertices, the color 1 is on an edge $(u,v)$ with $v\notin V_j$ (v can belong to $V_i$); while the color 2 is on an edge $(u,z)$ with $z\in V_j$;
\item there are $a_{ij}^{12}$ vertices in $V_i$, each with both colors 1 and 2 incident to it. For each vertex $u$ of these ones, there exist two edges $(u,v)$ and $(u,z)$ with both $v$ and $z$ in $V_j$. 	
\end{itemize} 

In other words, the set $A$ gives a more fine grained division of the vertices in $V$, based both on the labels and on the incident colors. Please note that we do not specify the vertices but just the numbers of them. For a given label $i\in [h]$, the above numbers are related to nine distinct subsets of $V_i$, which can also be empty.  Clearly, for each $i\in [h]$, it holds that  $\sum_{j\in[k]}(a_{i}^0+a_{ij^-}^1+a_{ij}^1+a_{ij^-}^2+a_{ij}^2+a_{ij^-}^{12}+a_{i1j}^{12}+a_{i2j}^{12}+a_{ij}^{12})=|V_i|$. 

It is also clear that, for a couple of disjoint labels $i,j\in[k]$, it holds that $a_{ij}^1+a_{i1j}^{12}+a_{ij}^{12}=a_{ji}^1+a_{j1i}^{12}+a_{ji}^{12}$. This is because both the left and the right parts of the equation equal the number of the edges from $V_i$ to $V_j$ colored with 1.  Analogously, for a couple of disjoint labels $i,j\in[h]$, it holds that $a_{ij}^2+a_{i1j}^{12}+a_{ij}^{12}=a_{ji}^2+a_{j2i}^{12}+a_{ji}^{12}$, which also equals the number of the edges from $V_i$ to $V_j$ colored with 2. Since $|A|= (8(h-1)+1)h=8h^2-7h$, and since each of this nonnegative integers is lover or equal to $|V|$, a simple upper bound to the number of ways of selecting the values in $A$ is $|V|^{(8h^2-7h-1)}$. This upper bound does not take into account how these values are related to each others.

We use $A$ in our dynamic programming algorithm, which runs on a given decomposition $\mathcal{X}$ of $G$. In particular, we use it in operation $iii)$. In the following we describe how to compute values $f(G,A)$ for each of the four operations of a given $\mathcal{X}$ decomposition of $G$.

\paragraph{Operation i)} For a given set $A$ on a graph $G$ obtained adding a new vertex $u_i$ to $G_1=(W,E_1)$, we want to compute $f(G,A)$ by using the already computed values $f(G_1,B)$. 

For a specific  set $A$, we have that 

\begin{itemize}
\item $f(G,A)=f(G_1,B)$ \quad if $a_i^0=b_i^0+1$, and $A\setminus \{a_i^0\}=B\setminus \{b_i^0\}$; 
\item $f(G,A)=-\infty$ \quad otherwise. 
\end{itemize}

In fact, since there are no new edges, $u_i$ is isolated, there are no feasible solutions with colors incident to $u_i$, and the maximum 2-edge colorable subgraph problem is the same as in $G_1$. This means that the only way to relate $A$ and $B$ with each other is to put $A\setminus \{a_i^0\}=B\setminus \{b_i^0\}$ and $a_i^0=b_i^0+1$.

The time complexity to compute  $f(G,A)$ for a  given set $A$ by using the already computed values $f(G_1,B)$ is $O(1)$, because there is a 
one-to-one relation between $A$ and $B$.   

\paragraph{Operation ii)} Here we want to compute $f(G,A)$ for a given $A$, where $G$ is obtained by joining $G_1$ and $G_2$, and by using the already computed values $f(G_1,B)$ and $f(G_2,C)$. Clearly, since there are no edges between $G_1$ and $G_2$, we can combine any two solutions for $G_1$ and $G_2$ just checking the compatibility among $A$, $B$, and $C$. This means that, for a given $A$, we have that 

\begin{itemize}
\item $f(G,A)=f(G_1,B)+f(G_2,C)$ \quad if $a_{i}^0=b_{i}^0+c_{i}^0$, $a_{ij^-}^1=b_{ij^-}^1+c_{ij^-}^1$, $a_{ij}^1=b_{ij}^1+c_{ij}^1$, $a_{ij^-}^2=b_{ij^-}^2+b_{ij^-}^2$, $a_{ij}^2=b_{ij}^2+c_{ij}^2$,
$a_{ij^-}^{12}=b_{ij^-}^{12}+c_{ij^-}^{12}$, 
$a_{i1j}^{12}=b_{i1j}^{12}+c_{i1j}^{12}$, $a_{i2j}^{12}=b_{i2j}^{12}+c_{i2j}^{12}$, $a_{ij}^{12}=b_{ij}^{12}+c_{ij}^{12}$; 
\item $f(G,A)=-\infty$ \quad otherwise. 
\end{itemize}


The time complexity for computing $f(G,A)$, for a given $A$ and by using the already computed values $f(G_1,B)$ and $f(G_2,C)$, is $O(|V|^{(8h^2-7h-1)})$, because we only need to choose the values in $B$.

\paragraph{Operation iii)} We exploit here the fact that our problem can be polynomially solvable on  bipartite graphs.

Let $G_1=(V,E_1)$ be a graph where apply operation iii); $i$ and $j$ the labels involved in it; $f(G_1, B)$ the already computed optimal values on $G_1=(V,E_1)$; $G=(V,E)$ the graph obtained after this operation; and $f(G, A)$ the relatives optimal values that we want to compute.

Let $E_1^{ij}$ be the subset of the edges in $E_1$ with one vertex in $V_i$ and the other in $V_j$. Let $G_1'=(V,E_1\setminus E_1^{ij})$ be the subgraph of $G_1$ after removing all the edges in $E_1^{ij}$. Let $(\text{\textbf{x}}',\text{\textbf{x}})$ be a solution for $G_1$ with value $f(G_1, B)$, where $\text{\textbf{x}}'$ is a vector of color variables $x_{uv}\in \{0,1,2\}$, one for each edge in the subgraph $G_1'$; while $\text{\textbf{x}}$ is a color variable vector of length $E_1^{ij}$. In other words, each variable $x_{uv}$ says if the edge $(u,v)$ is not colored, or it is colored with color 1 or 2. Please note that a solution $(\text{\textbf{x}}',\text{\textbf{x}})$ univocally identifies the vertices of $V_i$ and $V_j$ with respect to the integers in $B_{ij}$ and $B_{ji}$.  We call these subsets by using the same superscripts and subscripts in $b_{i}^0$, $\ldots$, $b_{ij}^{12}$ and in $b_{j}^0$, $\ldots$, $b_{ji}^{12}$, that is $V_{i}^0$, $\ldots$, $V_{ij}^{12}$, and $V_{j}^0$, $\ldots$, $V_{ji}^{12}$. Let $c\colon \{0,1,2,\}^n \rightarrow \mathbb{N}$ be a function, which returns the number of colored edges in a solution.


Let $H=(S,T,E_H)$ be the bipartite graph with the vertices divided into $S=(V_i^0\cup V_{ij^-}^1\cup V_{ij}^1\cup V_{ij^-}^2\cup V_{ij}^2)$ and $T=(V_j^0\cup V_{ji^-}^1\cup V_{ji}^1\cup V_{ji^-}^2\cup V_{ji}^2)$, and with $E_H\subseteq S\times T$. The edges in $E_H$ are the only ones we can use to create a feasible solution for $G$, adding colored edges to $(\text{\textbf{x}}',\text{\textbf{x}})$. In Figure~\ref{fig:graph_H} there is a scheme of $H$ that shows the connections among the sets of nodes. In particular, a normal arc in the scheme represents a complete bipartite subgraph in $H$, e.g., $(V_i^0,V_j^0)$ and $(V_{ij^-}^1,V_{ji^-}^1)$; and the two dotted arcs, $(V_{ij}^1,V_{ji}^1)$ and $(V_{ij}^2,V_{ji}^2)$, represent two general (maybe not complete) subgraphs. Clearly, in $E_H$, there are no edges in $(V_{ij^-}^1\cup V_{ij}^1)\times (V_{ij^-}^2\cup V_{ij}^2)$ and in $(V_{ij^-}^2\cup V_{ij}^2)\times (V_{ij^-}^1\cup V_{ij}^1)$, because there are no ways of color them with respect to $B$, and $\text{\textbf{x}}'$.

The edges in $E_H$ strictly depend on the solution part $\text{\textbf{x}}$. In order to maximise the number of  edges in $E_H$, we need to minimise the colored edges in  $ V_{ij}^1 \times  V_{ji}^1$ and in $ V_{ij}^2\times V_{ji}^2$. To achieve this goal, we can pass through a new solution $(\text{\textbf{x}}',\text{\textbf{y}})$ for $G$, which is still compatible with $B$, and with  $c(\text{\textbf{x}})=c(\text{\textbf{y}})$.

In $(\text{\textbf{x}}',\text{\textbf{y}})$ we color with 1 as many edges as possible from $V_{ij}^1$ to $V_{j1i}^{12}$ and to $V_{ji}^{12}$; and from $V_{j1}^1$ to $V_{i1j}^{12}$ and to $V_{ij}^{12}$. Then, the minimum number of edges colored with 1 from $V_{ij}^1$ to $V_{ji}^1$ is $b^1=|b_{ij}^1-(b_{j1i}^{12}+b_{ji}^{12})|=|b_{ji}^1-(b_{i1j}^{12}+b_{ij}^{12})|$. Analogously, the minimum number of edges colored with 2 from $V_{ij}^2$ to $V_{ji}^2$ is $b^2=|b_{ij}^2-(b_{j2i}^{12}+b_{ji}^{12})|=|b_{ji}^2-(b_{i2j}^{12}+b_{ij}^{12})|$. 

Going back to $E_H$, we select $b^1$ couples of nodes in $V_{ij}^1\times V_{ji}^1$, and we do not connect them in $H$, because they are connected via $b^1$ edges colored with 1 in $(\text{\textbf{x}}',\text{\textbf{x}})$.  Analogously, we remove $b^2$ distinct edges in $V_{ij}^2\times V_{ji}^2$, and we do not connect them in $H$, because they are connected via $b^1$ edges colored with 1 in $(\text{\textbf{x}}',\text{\textbf{x}})$

\begin{figure}[t]
\centering\includegraphics[scale=0.6]{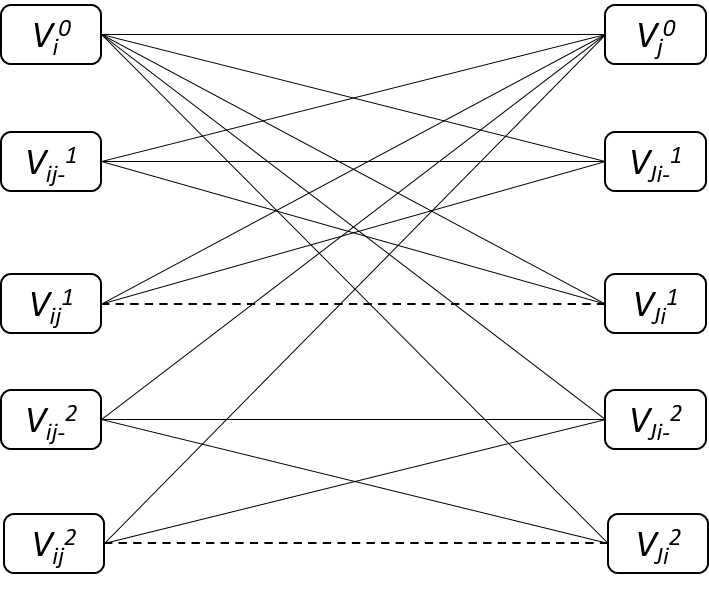}
\caption{Structure of graph $H$. Normal arcs state for complete bipartite subgraphs (e.g., $(V_i^0,V_j^0)$); dotted arcs state for bipartite subgraphs (e.g., $(V_{ij}^1,V_{ji}^1)$).}\label{fig:graph_H}
\end{figure}

We can now create a solution for $G$ coloring the edges in $H$, and adding them to $(\text{\textbf{x}}',\text{\textbf{y}})$. We can follow the simple steps blow.

\begin{itemize}
\item 1)  Select $b=\min\{b_i^0,b_j^0\}$ vertices in $V_i^0$ and  $b$ vertices in $V_j^0$. They form an even cycle with $2b$ edges, which can be colored alternating color 1 and color 2. Set $\overline{V}_j^0$ and $\overline{V}_i^0$ as the remaining vertices of $V_i^0$ and $V_j^0$, resp.. Clearly, either $\overline{V}_i^0=\emptyset$ or $\overline{V}_j^0=\emptyset$. Set $\overline{V}^0=\overline{V}_i^0\cup \overline{V}_j^0$.
\item 2) Find a maximum matching on the bipartite graph $H_1$ induced by $\overline{V}^0\cup V_{ij^-}^1\cup V_{ij}^1\cup V_{ji^-}^1\cup V_{ji}^1$. color the matching edges with 2.
\item 3) Find a maximum matching on the bipartite graph  $H_2$ induced by $\overline{V}^0\cup V_{ij^-}^2\cup V_{ij}^2\cup V_{ji^-}^2\cup V_{ji}^2$. color the matching edges with 1.
\end{itemize}

It can be easily shown that the maximum 2-edge-colorable subragh found following the three steps above for $H$ is optimum. 

Concluding, for a given set $A$, we can state that.

\begin{itemize}
\item $f(G,A)=f(G_1,B)+\max c(H)$ \quad if $A_{ij}$ and $A_{ji}$ derive from $B_{ij}$ and $B_{ji}$ following the solution obtained for $H$; and $A\setminus (A_i \cup A_j)=B\setminus (B_i \cup B_j)$;
\item $f(G,A)=-\infty$ \quad otherwise. 
\end{itemize}

We finally check that our method always leads to a solution with an optimal value $f(G,A)$. Suppose, for the sake of contradiction, that the optimal value $f(G,A)$ is greater than the number of colored edges in $(\text{\textbf{x}}',\text{\textbf{y}})$ plus $\max c(H)$, that is the value given by our procedure. Let $(\text{\textbf{z}}',\text{\textbf{z}})$ be the part of a $G$ optimal solution $f(G,A)$,  concerning the graph $G'_1$ and the edges in $E_1^{ij}$. The color vector $(\text{\textbf{z}}',\text{\textbf{z}})$, whatever it is, is also a feasible solution for $G_1$, and it is related to a set $B$ of values. In particular, it univocally identifies the sets $B_i$ and $B_j$, and the sets $V_{i}^0$, $\ldots$, $V_{ij}^{12}$, and $V_{j}^0$, $\ldots$, $V_{ji}^{12}$. Since, for a given $B$, the number of colored  edges in $E_1^{ij}$ is fixed, and equals the number of colored variables in  $\text{\textbf{z}}$, we can take a solution $(\text{\textbf{x}}',\text{\textbf{x}})$ with maximum value  $f(G_1,B)$, and apply our procedure, which build $(\text{\textbf{x}}',\text{\textbf{y}})$  and maximise the number of colored edges in $V_i\times V_j$. We get, then, a solution with value $f(G,A)$ ($\Rightarrow\Leftarrow$).



The time complexity for computing $f(G,A)$, for a given $A$, by using the already computed values $f(G_1,B)$ is $O(|V|^{(8k^2-7k-1)}(|V|+|E|))$, because we only need  to choose the values in $B$, $O(|V|^{(8k^2-7k-1)}$, and solve the two maximum matchings in $H_1$ and $H_2$ for each $B$, $O(|V|+|E|)$ \cite{tassa:2012}.  

\paragraph{Operation iv)} Let $G_1=(V,E_1)$ be a graph where renaming label $i$ with label $j$, $f(G_1, B)$ the already computed optimal values on $G_1=(V,E_1)$, $G=(V,E)$ the graph obtained after this operation, and $f(G, A)$ the relatives optimal values that we want to compute.

Clearly, since $G$ and $G_1$ are equal to each other, apart from the labels, any \textit{maximum 2-edge-colorable subgraph} of $G_1$ is also a \textit{maximum 2-edge-colorable subgraph} of $G$ and vice versa. This means that, for a given $A$, we have that.

\begin{itemize}
\item $f(G,A)=f(G_1,B)$ \quad if
\begin{itemize}
\item 	$A_i=\{0,\ldots,0\}$; $A_{lk}=B_{lk}$ for every disjoint labels $l,k\in ([h]\setminus \{i,j\})$; 
\item $a_{li^-}^1=b_{li^-}^1+b_{li}^1$, $a_{li}^1=0$, $a_{li^-}^2=b_{li^-}^2+b_{li}^2$, $a_{li}^2=0$, $a_{li^-}^{12}=b_{li^-}^{12}+b_{l1i}^{12}+b_{l2i}^{12}+b_{li}^{12}$, $a_{l1i}^{12}=0$, $a_{l2i}^{12}=0$, $a_{li}^{12}=0$, for every $l\in ([h]\setminus \{i,j\})$;
\item $a_{lj^-}^1=b_{lj^-}^1-b_{li}^1$, $a_{lj}^1=b_{lj}^1+b_{li}^1$, $a_{lj^-}^2=b_{lj^-}^2-b_{li}^2$, $a_{lj}^2=b_{lj}^2+b_{li}^2$, 

$a_{lj^-}^{12}=b_{lj^-}^{12}-b_{l1i}^{12}-b_{l2i}^{12}-b_{li}^{12}$, 

$a_{l1j}^{12}=b_{l1j}^{12}+b_{l1i}^{12}$, $a_{l2j}^{12}=b_{l2j}^{12}+b_{l2i}^{12}$, 

$a_{lj}^{12}=b_{lj}^{12}+b_{li}^{12}$, for every $l\in ([h]\setminus \{i,j\})$;

\item $a_{jl^-}^1=b_{lj^-}^1-b_{li}^1$, $a_{lj}^1=b_{lj}^1+b_{li}^1$, $a_{lj^-}^2=b_{lj^-}^2-b_{li}^2$, $a_{lj}^2=b_{lj}^2+b_{li}^2$, 

$a_{lj^-}^{12}=b_{lj^-}^{12}-b_{l1j}^{12}-b_{l2j}^{12}-b_{lj}^{12}$, 

$a_{l1j}^{12}=b_{l1j}^{12}+b_{l1i}^{12}$, $a_{l2j}^{12}=b_{l2j}^{12}+b_{l2i}^{12}$, 

$a_{lj}^{12}=b_{lj}^{12}+b_{li}^{12}$, for every $l\in ([h]\setminus \{i,j\})$;

\end{itemize}

for every $l\in ([h]\setminus \{i,j\})$; $A_{il}=B_{il}+B_{jl}$ for every $l\in ([h]\setminus \{i,j\})$;
\item $f(G,r,\mathcal{V},A)=-\infty$ \quad otherwise. 
\end{itemize}

In fact, since $G_1$ and $G_2$ are equivalent with each other in the structure, we just need to check if $A$ and $B$ are compatible.

The time complexity for computing $f(G,A)$, for a given $A$, by using the already computed values $f(G_1,B)$ is $O(1)$, because there is a one-to-one relation between $A$ and $B$.  

We can finally write the following theorem.

\begin{theorem}\label{thm:CliqueWidth}
Given an instance of \textit{maximum 2-edge-colorable subgraph problem} with a graph $G$, and a clique decomposition $\mathcal{X}$ of width $h$ for $G$, the \textit{maximum 2-edge-colorable subgraph problem} is solvable in $O(h^2 |V|^{(16h^2-14h-1)}(|V|+|E|))$ time.
\end{theorem}
\begin{proof}
Since the length of $\mathcal{X}$ is at most $O(h^2|V|)$ \cite{courcelle:2000}, and the most expensive operation, iii), costs $O(|V|^{(16h^2-14h-2)}(|V|+|E|))$, the time complexity of our dynamic programming algorithm is  $O(h^2|V|^{(16h^2-14h-1)}(|V|+|E|))$.   	
\end{proof}


\section{The maximum 2-edge-colorable subgraph problem and the method of iterative compression}
\label{sec:IterativeCompression}

In this section, we consider two problems related to the maximum 2-edge-colorable subgraph problem. Using the method of iterative compression (see Section 4 of \cite{FPTbook}), we show that these two problems are FPT with respect to the budget $k$.

The first problem that we will consider is the following:
\begin{problem}\label{prob:Removekedges}
	Given a graph $G$ and an integer $k$, is there $X\subseteq E(G)$, such that $|X|\leq k$ and $G-X$ is $2$-edge-colorable?
\end{problem}

Observe that if $G$ is a cubic graph, then in order to get a 2-edge-colorable subgraph, for each vertex $v$ of $G$ we have to remove at least one edge incident to $v$. Thus, we have to remove at least $\frac{|V|}{2}$ edges. Now, observe that in cubic graphs there is a set of size $\frac{|V|}{2}$ whose removal leaves a 2-edge-colorable subgraph if and only if $G$ is 3-edge-colorable. Thus, combined with \cite{holyer:1981}, we have that Problem \ref{prob:Removekedges} is NP-complete even for cubic graphs.

We continue with the following lemma.
\begin{lemma}
	\label{lem:MaxDegreetwoAuxProblem} Consider the following decision problem: given a bipartite graph $H$ with $\Delta(H)\leq 2$, two subsets of edges $E_1, E_2\subseteq E(H)$, and an integer $k$. The goal is to check whether there is a subset $X\subseteq E(H)$, such that $|X|\leq k$ and $H-X$ has a proper 2-edge-coloring $f:E(H)\backslash X\rightarrow \{1, 2\}$, such that if $e\in E_1-X$, then $f(e)=1$, and $e\in E_2-X$, then $f(e)=2$. There is a polynomial time algorithm for solving this problem.
\end{lemma}

\begin{proof} We will actually prove that the minimization version of this problem is polynomial time solvable. Thus, we can find the smallest size of a feasible set $X$ can compare it with $k$. Clearly, when solving the minimization problem we can focus solely on connected graphs. Thus, $H$ is a path or a cycle. Let us start with paths. Let $E_1\cup E_2=\{e_1,....,e_r\}$. Assume that the labelling is done so that when you look at the path from left to right, the edges appear in this order. Let us define the notion of a conflict. A pair of consecutive edges $e_j, e_{j+1}$ forms a conflict if the length of the path between them is even and they belong to the same $E_1$ (or the same $E_2$) (roughly speaking their colors should be the same) or they belong to different $E_j$s and their distance is odd. Observe that this definition is meaningful even when $e_j$ and $e_{j+1}$ are incident to the same vertex. In this case the distance is zero, hence we have a conflict if they must have the same color.
	
	Now the critical observation is that if one has a smallest $X$ that overcomes the conflicts, then we can always assume that $X\subseteq \{e_1,....,e_r\}$. This follows from the observation that if you have removed an edge $f$ overcoming a certain conflict, then you can remove the closest right edge from $\{e_1,....,e_r\}$. Observe that in the latter case this edge from $E_1\cup E_2$ will remove the conflict next to it, too. Thus, if we assume that $e_{i_1}, e_{i_1+1} $,...,$e_{i_q}, e_{i_q+1}$ form consecutive conflicts, then by removing $X=\{e_{i_1+1},e_{i_2+1},...,e_{i_q+1}\}$ we will get rid of all conflicts. Moreover, the feasible set $X$ will be the smallest.
	
	Now, let us consider the case of cycles $C$. Again, we can assume that $E_1\cup E_2=\{e_1,....,e_r\}$. Moreover, we will assume the same way of labelling the edges. The conflicts will be defined in the same way (we assume some circumference order on the cycle $C$). Now, let us show that any two consecutive edges $e$ and $e'$ from $E_1\cup E_2$ must form a conflict. Assume not. Let $e$ and $e'$ be two consecutive edges such that there is no conflict between them. Consider the subpath of the cycle $C$ starting from $e$ and ending on $e'$. Solve the optimization problem in this subpath $P$. Clearly we can extend the 2-edge-coloring of $P-X$ to that of $C-X$. 
	
	
	Thus, we are left with the assumption that $e_1, e_2$, ..., $e_{r-1},e_r$ and $e_r, e_1$ form conflicts. Now, if $r$ is even then by taking $X=\{e_2, e_4,..., e_r\}$ we will have that $G-X$ is without conflicts and clearly $X$ is smallest. On the other hand, if $r$ is odd, then $X=\{e_2, e_4,..., e_{r-1}\}\cup \{e_r\}$ is a smallest feasible set. The proof is complete.
\end{proof}

%

As in \cite{FPTbook} (see Section 4), this lemma implies that the disjoint version of the problem is FPT with respect to $k$. 
\begin{lemma}
	\label{lem:Disjoint2coloring} In the disjoint version of the problem we are given a graph $G$, integer $k$ and $W\subseteq E(G)$, such that $G-W$ is 2-edge-colorable and $|W|=k+1$. The goal is to check whether there is $X\subseteq E(G)\backslash W$ such that $|X|\leq k$ and $G-X$ is 2-edge-colorable. This problem can be solved in time $2^k\cdot poly(size)$.
\end{lemma}

\begin{proof} The proof is similar to the ones given in Section 4 of \cite{FPTbook}. Since the edges of $W$ cannot deleted in $G-X$, it is necessary that $G[W]$ is 2-edge-colorable. Let us consider all possible 2-edge-colorings $f_W$ of $W$. Clearly, their number is at most $2^{|W|}=2^{k+1}$. Now for each of them let $E_1^W$ and $E_2^W$ be the color classes of $f_W$. Let $E_1$ be the set of edges of $G-W$ that are adjacent to an edge from $E_2^W$. Similarly, let $E_2$ be the set of edges of $G-W$ that are adjacent to an edge from $E_1^W$. Observe that any edge of $E_1$ either has to be deleted or colored with 1. Similarly, any edge of $E_2$ either has to be deleted or colored with 2. Thus, in order to solve this problem we have to solve the instance of the problem from Lemma \ref{lem:MaxDegreetwoAuxProblem} for $(G-W, E_1, E_2, k)$. Observe that the edges outside $W$ which are incident to degree-two vertices of $G[W]$, must be deleted (they must be taken in $X$). Also observe that by the definition of $W$, $G-W$ satisfies the conditions of the Lemma \ref{lem:MaxDegreetwoAuxProblem}. According to the lemma, each of these $2^{|W|}=2^{k+1}$ instances can be solved in time $ poly(size)$. Thus, we have the desired running time. The proof is complete.
\end{proof}

Since the disjoint version of our problem can be solved in time $2^k\cdot poly(size)$, we immediately have the following result as a consequence of the method of iterative compression (Section 4 of \cite{FPTbook}):
\begin{theorem}
	\label{thm:RemovedEdges} Problem \ref{prob:Removekedges} is FPT with respect to $k$ and it can be solved in time $3^k\cdot poly(size)$.
\end{theorem}

\begin{proof} This just follows from the method of iterative compression. See Section 4.1.1 of \cite{FPTbook} The proof is complete.
\end{proof}

%

Now, we turn to the vertex-set removal version of the problem. 
\begin{problem}\label{prob:Removekvertices}
	Given a graph $G$ and an integer $k$, is there $X\subseteq V(G)$, such that $|X|\leq k$ and $G-X$ is $2$-edge-colorable?
\end{problem}

It can be shown that Problem \ref{prob:Removekvertices} is NP-complete, too. This just follows from Lemma 1,2 and Theorem 1 from \cite{Chiarelli:2018}. The key observation from Theorem 1 there is that when the authors find an odd cycle transversal (a subset $V'\subseteq V(H)$, such that $H-V'$ is bipartite) of the line graph $L(G)$, they actually have that it is also gives an even 2-factor. In other words, $L(G)-V'$ is 2-edge-colorable.

We continue with the following lemma.
\begin{lemma}
	\label{lem:Degree3constraints} Let $K$ be a graph with $\Delta(K)\leq 2$. Consider a graph $H$ obtained from $K$ by attaching maximum one pendant edge to some degree-two vertices of $K$. Assume that the resulting degree three vertices are independent, moreover, on pendant edges we have a color from $\{1,2\}$ which must be satisfied. Consider the following problem: find a smallest subset $J$ of degree two vertices of $H$ (we are not allowed to take degree-one or degree-three vertices), such that $H-J$ admits a 2-edge-coloring respecting the constraints on pendant edges. This problem can be solved in polynomial time.
\end{lemma}

\begin{proof} Since we are solving a minimization problem, we can assume that $H$ is connected. Let us start with the case when $K$ is a path. Look at the path from left to right, and take the first two degree one vertices with constraints. Let $w_1$ and $w_2$ be these two degree one vertices. First, assume that the unique neighbor of $w_1$ is of degree two. Then consider the vertices $u,v$, the neighbors of the degree three vertex adjacent to $w_2$. Assume $u$ is between $w_1$ and $w_2$. Observe that we have to remove at least one of $u$ or $v$. Now, if the path between $w_1$ and $w_2$ does not create a conflict (see the proof of Lemma \ref{lem:MaxDegreetwoAuxProblem} for the definition of a conflict), then clearly we can just remove $v$ and solve the resulting smaller instance. On the other hand, if the path between $w_1$ and $w_2$ creates a conflict, then we need to remove a vertex between them, thus it is safe to remove $u$ and solve the remaining smaller instance. Now, assume that the neighbor of $w_1$ is of degree three. Then on its left there is no other conflicting edge. Thus, it suffices to remove the neighbor of the neighbor of $w_1$ that is between $w_1$ and $w_2$, and solve the resulting smaller instance. This allows us to consecutively solve the case of paths.
	
	Now, assume that we have a cycle in $K$ and some pendant edges with constraints are added to it in order to obtain $H$. If there are two consecutive degree one vertices such that the path between them is not creating a conflict, then let this path be $P$. Observe that $|V(P)|\geq 3$. First assume that $P$ contains at least four vertices. Consider the unique neighbors of neighbors of $w_1$ and $w_2$, respectively, that lie outside $P$. Clearly we can remove these two vertices (because we have to remove at least one around $w_1$ and $w_2$), and solve the resulting problem for the resulting path, and of course we can find the coloring of $P$ satisfying the constraints because it is conflict free. On the other hand, if $|V(P)|=3$, we solve two cases of the path problem: first we remove the unique neighbor of degree-three vertices (the unique degree-two vertex of $P$) and find the smallest set of vertices for the remaining path problem. Next, we remove the two vertices adjacent to the degree-three vertices that differ from the common vertex between them and solve the resulting path problem. Then we take the smaller of the two. Observe that the path instances are not reduced to multiple instances. Hence we get just two instances here.
	
	Thus we are left with the case, that any two consecutive pendant edges of the cycle form a conflict. Now for each of the pendant edges, remove, for example, the right degree two vertex adjacent to the degree vertex. Clearly this will be a smallest set, as otherwise, if we assume that there is a smaller one, then clearly there are two consecutive pendant edges such that between them no vertex is removed, hence there can be no 2-coloring extending the constraints. The proof is complete.
\end{proof}

Our next lemma works with the following extension of the previous problem.
\begin{lemma}
	\label{lem:ManyPendantEdges} Let $K$ be a graph with $\Delta(K)\leq 2$. Consider a graph $H$ obtained from $K$ by adding new vertices $w$, and joining $w$s to some vertices of $K$ with edges and adding a color from $\{1,2\}$ as a constraint. All edges adjacent to the same $w$ have the same color as a constraint. Consider the following decision problem: for this type of graph $H$ and an integer $k$, check whether there is a subset $X\subseteq V(K)$ (we are not allowed to take the vertices $w$ in $X$), such that $|X|\leq k$ and $G-X$ admits a 2-edge-coloring that satisfies the constraints on edges incident to $w$. This problem is FPT with respect to $k$ and it can be solved in time $2^{k(k+1)}\cdot poly(size)$.
\end{lemma}

\begin{proof} Let $Q_1$ be the number of those $w$s that are adjacent to exactly one edge. Similarly, let $Q_{\geq 2}$ be the number of those $w$s that are adjacent to at least $2$ edges. Observe that we can assume that $Q_{\geq 2}\leq k$ as for each such vertex we have to remove at least one neighbor from $K$, thus if their total number is greater than $k$, the instance is a no-instance. Thus, $Q_{\geq 2}\leq k$. Now observe that each fixed $w$ of degree at least two, must have degree at most $k+1$, as if its degree is at least $k+2$, then at least $k+1$ neighbors should be removed, hence we have a no-instance. Let us call these neighbors of $w$s as roots. Thus, each $w$ is adjacent to at most $k+1$ roots. Hence, the total number of roots is at most $Q_1+k(k+1)$.
	
	Now let us guess all subsets of those $k(k+1)$ roots that are not counted in $Q_1$. Their number is at most $2^{k(k+1)}$. Let $R$ be such a guess. Then in the graph $H-R$ we need to check whether $d_{H-R}(w)\leq 1$ for any $w$. Also we need to have $|R|\leq k$. If one of these conditions is not satisfied then the guess is wrong. Now, if these conditions are satisfied then clearly we cannot have adjacent degree three vertices, as in the solution at least one of adjacent degree three vertices must be removed, and hence $R$ is not the correct guess. Thus, if this condition is also satisfied for $R$, we need to solve the instance of the problem from previous lemma and check whether a smallest subset of size at most $k'=k-|R|$ exists. Lemma \ref{lem:Degree3constraints} guarantees that each of these instances can be solved in polynomial time. Thus, the total running-time of our algorithm will be $2^{k(k+1)}\cdot poly(size)$. The proof is complete.
\end{proof}

Now, we solve the disjoint version of our problem.
\begin{lemma}
	\label{lem:DisjointVertexProblem} In the disjoint version of our problem, we are given a graph $G$, integer $k$ and $W\subseteq V(G)$, such that $G-W$ is 2-edge-colorable and $|W|=k+1$. The goal is to check whether there is $X\subseteq V(G)\backslash W$ such that $|X|\leq k$ and $G-X$ is 2-edge-colorable. This problem can be solved in time $2^{(k+1)^2}\cdot poly(size)$.
\end{lemma}

\begin{proof} Observe that if the solution exists, $G[W]$ must be 2-edge-colorable. Thus, we can guess all its 2-edge-colorings. Their total number is at most $2^{k+1}$. Now, for each of those guesses $f_W$, we define the set $E_1$ as those edges that are adjacent to an edge of $W$ of color 2, and similarly, let $E_2$ be those edges of $G$ that are adjacent to an edge $W$ of color 1. If an edge is both from $E_1$ and $E_2$ then we must delete its neighbor outside $W$. Now, in order to answer our problem, we need to solve the instance of the problem from the previous lemma for the graph $G-W$ with constraints on $E_1$ and $E_2$ and the parameter $k$. Observe that we are not allowed to touch the vertices in $W$. According to the previous lemma, each such instance can be solved in time $2^{k(k+1)}\cdot poly(size)$. Thus the total running time for the disjoint version of our problem is 
	\[2^{k+1}\cdot 2^{k(k+1)}\cdot poly(size)=2^{(k+1)^2}\cdot poly(size).\]
	The proof is complete.
\end{proof}

Since the disjoint 2-vertex-coloring problem can be solved in time $2^{(k+1)^2}\cdot poly(size)$, we immediately have the following result as a consequence of the method of iterative compression:
\begin{theorem}
	\label{thm:RemovedVertices} Problem \ref{prob:Removekvertices} is FPT with respect to $k$ and can be solved in time $2^{(k+1)^2+k}\cdot poly(size)$.
\end{theorem}

\begin{proof} This just follows from the method of iterative compression (see Section 4.1.1 of \cite{FPTbook}). The actual expression that needs to be bounded is the following one multiplied with $poly(size)$:
	\[\sum_{i=0}^{k}\binom{k+1}{i}2^{(k-i+1)^2}\leq 2^{(k+1)^2}\cdot \sum_{i=0}^{k}\binom{k+1}{i}\leq 2^{(k+1)^2}\cdot 2^{k+1}=2^{(k+1)^2+k+1}.\]
	Thus, the running time will be $2^{(k+1)^2+k}\cdot poly(size)$. The proof is complete.
\end{proof}

\section{Conclusion and future work}
\label{sec:conc}

In this paper, we considered the problem of assigning tasks to agents under time conflicts, which can also be applied to frequency allocation in point-to-point wireless networks. The problems is interesting both from the application point of view and from the theoretical one. Since it is our first approach to it, we focused on a restricted version of the problem, that remains computationally difficult. After collocating it in the world of task assignment problems, we formulate it as an edge-coloring problem, and we focused on obtaining parameterized complexity results. 

In particular, we formulated it as the maximum 2-edge-colorable subgraph problem. Our results state that this problem is paraNP-hard with respect to radius, diameter and $|V|-MaxLeaf$. On the positive side, it is fixed-parameter tractable with respect to $|V|-\delta$, branchwidth, treewidth, the dimension of the cycle space and $MaxLeaf$. Moreover, it is polynomial time solvable for the graphs of bounded cliquewidth.

From our perspective the following line of research is suitable for future research. For a graph $G$, let $\tau(G)$ be the size of the smallest vertex cover of $G$. Since in any graph $\nu(G)\leq \tau(G)$, Corollary \ref{cor:FPTnu1} and Lemma \ref{lem:Redparam} imply that the maximum 2-edge-colorable subgraph problem is FPT with respect to $\tau(G)$. We would like to ask:

\begin{question}
	\label{que:VertexCoverMatching} Is the maximum $2$-edge-colorable subgraph problem FPT with respect to $\tau(G)-\nu(G)$?
\end{question} The classical 2-approximation algorithm for the vertex cover problem and its analysis imply that for any graph $G$, we have $\tau(G)\leq 2\cdot \nu(G)$. This inequality means that in any graph $G$, $\tau(G)-\nu(G)\leq \nu(G)\leq \tau(G)$. Thus, a positive answer to Question \ref{que:VertexCoverMatching} will strengthen Corollary \ref{cor:FPTnu1} and its consequence for $\tau(G)$. 



Some lines of research that may be worth to investigate concern the generalisation to the weighted case, where each color can have different weights even in combination with specific edges, adding different classes of constraints among colors, budged constraints on the sum of the weights like in \cite{aloisio_a:2020,aloisio:2020}, and analyzing them with respect to different graph topologies, like it is done in \cite{alo_nav:2019,aloisio:2019}. A next interesting step would be to investigate the problem on hypergraphs, when a task has to be carried out by more than two agents. 

\section*{Acknowledgement} The authors would like to thank professor Michele Flammini for his attention to this work. The second author would like to thank Dr. Zhora Nikoghosyan for useful discussions on Hamiltonian graphs. He also would like to thank Dr. Kenta Ozeki for pointing him out the paper \cite{Ding}.



\end{document}